\documentclass[11pt]{article}

\usepackage[letterpaper,margin=1in]{geometry}
\usepackage{amsmath,amssymb,amsthm,mathtools}
\usepackage{microtype}

\numberwithin{equation}{section}
\usepackage{xcolor}
\usepackage{booktabs}
\usepackage{tabularx}
\usepackage{algorithm}
\usepackage[noend]{algpseudocode}
\usepackage{tikz}
\usetikzlibrary{arrows.meta,calc,fit,positioning}
\usepackage[
  colorlinks=true,
  linkcolor=blue!55!black,
  citecolor=blue!55!black,
  urlcolor=blue!55!black
]{hyperref}
\usepackage{authblk}

\newtheorem{theorem}{Theorem}[section]
\newtheorem{lemma}[theorem]{Lemma}

\newtheorem{corollary}[theorem]{Corollary}
\theoremstyle{definition}

\newtheorem{remark}[theorem]{Remark}

\newcommand{\ALG}{\operatorname{ALG}}
\newcommand{\OPT}{\operatorname{OPT}}
\newcommand{\BP}{\operatorname{BP}}

\newcommand{\E}{\mathbb{E}}
\newcommand{\R}{\mathbb{R}}
\newcommand{\vol}{\operatorname{vol}}
\newcommand{\spanof}{\operatorname{span}}

\title{Online Bin Packing with Per-Bin Maximum Delay}
\author{
Tianhang Lu,
Runtian Ren,
Shengcai Liu
}
\affil{
Guangdong Provincial Key Laboratory of Brain-Inspired Intelligent Computation,\\
Department of Computer Science and Engineering,\\
Southern University of Science and Technology, Shenzhen 518055, China\\
liusc3@sustech.edu.cn
}
\date{}

\begin{document}

\pagenumbering{roman}
\maketitle
\thispagestyle{empty}

\begin{abstract}
In online bin packing with delay, items of sizes in $(0,1]$ arrive over
time and may wait before being packed into unit-capacity bins. A service
selects a nonempty feasible batch of pending items and seals it in a new
bin, at a service cost of one. We study the objective that adds one
maximum waiting time for each sealed bin to the total service cost.
Thus a bin's delay is determined by its oldest item, and younger items
that fit can join it without increasing that delay.
We give a randomized online algorithm with strict competitive ratio
$2$ against an oblivious adversary, running in
polynomial time in the ideal random-real model. It maintains one Next-Fit
buffer and uses a single random silence threshold. Its analysis combines
an offline-bin fragmentation bound with a two-unit saving per silence
cluster. The ratio is tight for this policy; $b$ fair bits suffice for
ratio $2+2^{-b}$. A fixed threshold gives deterministic ratio $1+\sqrt3$.
The lower bounds $2$ and $e/(e-1)$ are attained when capacity is nonbinding.
Offline, the problem minimizes the number of bins plus their total arrival
span. We prove strong NP-hardness, an absolute approximation barrier of
$3/2$, a deterministic polynomial $2$-approximation, and exact algorithms
for several special cases. We also give stronger online ratios with exact
endpoint-packing oracles and an AFPTAS for the fixed-length weighted
endpoint subproblem. Finally, for Poisson arrivals of half-size items, we
determine the offline and optimal causal long-run cost rates. Two natural
asymptotic ratio measures coincide and have maximum $4/3$.
\end{abstract}

\clearpage
\pagenumbering{arabic}
\section{Introduction}
\label{sec:introduction}

Online bin packing with delay asks when to dispatch items as well as how
to pack them. Items arrive over time, and holding an item may allow it to
share a unit-capacity bin with later arrivals. A dispatched bin incurs a
service cost, while waiting incurs delay. Classical formulations charge
delay separately to every item
\cite{ahlroth2013,azar2019clustering,epstein2021}. Their tradeoff is between
the number of bins and the total waiting accumulated by all items.

Total waiting is a natural measure when each item incurs an independent
holding charge. It measures a different quantity when a dispatched batch
is assessed by its longest wait. In that case, a younger item that fits in
an existing batch can share its delay charge. Our purpose is to understand
how this change interacts with the defining constraint of bin packing:
even if waiting makes a large batch attractive, that batch may not fit.
The temporal incentive to combine items and the spatial constraint that
separates them must be analyzed together.

We study \emph{online bin packing with per-bin maximum delay}. If a
nonempty feasible bin $B$ is sealed at time $t$, it contributes
$1+\max_{i\in B}(t-a_i)$ to the objective, where $a_i$ is item $i$'s
arrival time. The first term is the bin's service cost; the second is its
oldest item's waiting time. The objective sums these contributions over
all bins. Service is instantaneous, and every item must eventually be
packed, without notification that the input has ended.

How much of the simple maximum-delay batching picture survives a capacity
constraint? Can an efficient online algorithm exploit the shared delay
without repeatedly solving a hard packing problem? And how far are such
algorithms from the guarantees possible with exact packing information?

When all items fit in one bin, this is per-batch maximum-delay TCP
acknowledgment. The deterministic optimum is two~\cite{dooly2001tcp},
and we give a self-contained randomized optimum of $e/(e-1)$ against an
oblivious adversary. Capacity changes the structure of the comparison.
An offline bin can skip intervening arrivals, and the arrival spans of
different bins can overlap. In particular, a partition into consecutive
time batches need not describe an offline optimum. Our main algorithm
therefore uses consecutive groups only to organize its own execution;
its analysis compares them with an arbitrary feasible offline packing.

\paragraph{Our main contribution and its proof idea.}

Our Capacity-Overflow Random-Silence algorithm (CORS) maintains one
Next-Fit buffer. It samples one $X\in[0,1]$ uniformly.
Every fitting arrival enters the buffer and resets its timer to $X$ time
units after that arrival. A nonfitting arrival seals the old buffer and
starts a new one; $X$ units of silence seal the current buffer.
The same sample is used throughout the input.

CORS is strictly $2$-competitive against an
oblivious adversary. It runs in polynomial time in the ideal random-real
model. Its factor is tight for the fixed-order buffer. With $b$ random bits its factor is at most $2+2^{-b}$. A fixed silence threshold gives a strict $(1+\sqrt3)$ guarantee in
polynomial time. The unrestricted deterministic and randomized lower bounds
are $2$ and $e/(e-1)$, respectively, and both are attained when capacity
is nonbinding.

The CORS proof retains the saving created by each silence event. For a
fixed threshold, let $C_X$ denote the cost of batching the arrivals while
ignoring capacity, $F_X$ the number of overflow closures, and $S$ the
number of silence clusters. The delays telescope within each cluster,
giving $\ALG_X=F_X+C_X$. Next Fit uses at most $2\BP(J)-1$ bins on
cluster $J$. If $H_X$ counts the fragments of a fixed offline packing
cut by the silence clusters, then
\[
 F_X\le2H_X-2S,
 \qquad \ALG_X\le2H_X+(C_X-2S).
\]
A gap is cut with probability at most its length, implying
$\E[H_X]\le\OPT$. A direct gap calculation gives
$\E[C_X-2S]\le0$: the initial term is $-1/2$, and each consecutive gap of length $d$ contributes $-\min\{d,1\}^2/2$.
Adding these bounds proves the competitive ratio. Keeping the negative
term $-2S$ is essential to this calculation; separately bounding openings
and delay gives a larger constant.

\paragraph{Offline structure and other online guarantees.}
An offline optimum seals each bin at its last assigned arrival. Thus
\[
 \OPT=\min_{\mathcal P}\sum_{B\in\mathcal P}(1+\spanof(B)),
 \qquad \spanof(B)=\max_{i\in B}a_i-\min_{i\in B}a_i,
\]
where $\mathcal P$ ranges over capacity-feasible partitions. This formula
gives strong NP-hardness and rules out absolute approximation below $3/2$
unless $\mathrm P=\mathrm{NP}$. We obtain a polynomial
$2$-approximation by enumerating silence thresholds, exact polynomial algorithms for equal sizes,
nonbinding capacity, and bin cardinality at most two, and an $O(3^n)$
algorithm for arbitrary inputs.

Periodic batching gives stronger competitive ratios when exact endpoint
packing is available. A fixed grid attains $(3+\sqrt5)/2$, and a uniformly
shifted grid attains $1+1/\sqrt2$. These algorithms may solve a strongly
NP-hard subproblem at each endpoint. We give polynomial implementations
using an age-ordered ordinary packing algorithm, and an AFPTAS for the
weighted endpoint problem at any fixed rational phase length $L$.
The AFPTAS has an instance-independent additive term
$\operatorname{poly}_L(1/\varepsilon)$; charging that term in every phase
does not establish the unrestricted ratio for a polynomial online algorithm.

\paragraph{An exact stochastic case.}
For rate-$\lambda$ Poisson arrivals of items of size $1/2$, we determine
both the offline and the optimal causal long-run cost rates. An optimal
offline solution can pair adjacent arrivals, yielding the recurrence
$C_i=\min\{C_{i-1}+1,C_{i-2}+1+a_i-a_{i-1}\}$.
Its stationary marginal-saving process gives offline rate
$\lambda(\lambda+2)/(2(\lambda+1))$.

Online, a causal exchange allows two pending items to be paired immediately.
The remaining decision is a timeout for a single item. Immediate service is
optimal when $\lambda\le1$. For $\lambda>1$, finite timeouts
$\tau_j=(2/\lambda)\log(j+1)$ attain rate $(\lambda+1)/2$ while still
clearing the last item without an end signal. A conditional cycle-cost
inequality proves optimality among all admissible causal policies.
The ratio of expectations and the expected realized ratio, conditional on
nonempty input, converge to the same value, at most $4/3$.
These are long-run stochastic guarantees, distinct from the adversarial
competitive ratios above.

\paragraph{A shared clock, constrained by capacity.}
For a fixed batch of $q$ pending items, additive delay grows at rate $q$,
whereas maximum delay grows at rate one. This is the temporal simplification
behind the algorithm: the gap since the last arrival can determine when
to stop waiting for another item that might share service. The silence
timer is a decision rule; the actual delay is still measured from the
oldest item in the bin. These are different clocks, and their relationship
is captured by the telescoping cost within a silence cluster.

Capacity prevents us from clearing an arbitrarily large cluster in one
bin. Overflow creates additional purchases, but it also creates a useful
saving in the analysis: the last bin in every cluster is already counted
by the uncapacitated service cost. Retaining this saving is what lets the
proof combine the temporal and packing bounds effectively. The comparison must account for timing and capacity jointly; treating either component in isolation loses the saving.

\paragraph{Related work.}
The classical bin-packing literature distinguishes absolute from asymptotic
guarantees.  The optimal absolute competitive ratio for deterministic
ordinary online bin packing is $5/3$~\cite{balogh2019}; First Fit Decreasing
has absolute approximation ratio $3/2$, which is best possible unless
$\mathrm P=\mathrm{NP}$~\cite{simchilevi1994}.  Configuration-LP pricing and
linear grouping originate in the classical asymptotic schemes for bin
packing~\cite{karmarkar1982}; we use a knapsack FPTAS~\cite{lawler1979} to
preserve the additional nested age constraints.  Ahlroth et al. introduced online bin packing with item-additive delay and holding
costs~\cite{ahlroth2013}.  Azar et al. connected additive-delay bin packing
to the price of clustering~\cite{azar2019clustering}, and Epstein improved
the resulting competitive ratio~\cite{epstein2021}.  Our delay is not
request-additive, and therefore those clustering charges do not transfer.

When capacity never binds, our problem is exactly the per-batch
maximum-delay TCP objective considered by Dooly et al.; their
Greedy result gives the deterministic factor two~\cite{dooly2001tcp}.
Bhore et al. recently revisited such batch-aware objectives
within a general delay-cost framework~\cite{bhore2026general}.  The
distribution $e^x/(e-1)$ and the constant $e/(e-1)$ are classical in
randomized online batching~\cite{karlin2001dynamic,seiden2000guessing}.  We
give a self-contained exact randomized frontier for this nonbinding class,
then show how much of it survives under bin capacity.

Ren and Liu's companion manuscript \emph{Online Network Design with Batch
Processing} charges job flow time and holds the server for the duration
of the longest job in a batch. Bin capacity here restricts which items
can leave together, even when arbitrarily many bins can seal at one time.

Capacitated line aggregation with maximum delay~\cite{lu2026capacitated}
includes our spatial model by placing every item at distance one and using
its indivisible-size extension. Its path-load algorithm coordinates
different radii and also supports capacity-feasible online advice. Our
single-location result gives a sharper randomized factor through silence
cuts and offline-bin fragmentation; the two guarantees are not interchangeable.

\paragraph{Organization.}
Section~\ref{sec:model} gives the model and computation conventions.
Section~\ref{sec:offline} derives the offline span objective.
Sections~\ref{sec:randomized-event} and~\ref{sec:fixed-silence} develop
the efficient online algorithms. Section~\ref{sec:lower-bounds} isolates
the temporal lower bounds, and Section~\ref{sec:further-overview}
explains the oracle, approximation, and stochastic results. Complete
phase, endpoint-packing, and stochastic analyses follow in the appendices.

\section{Preliminaries}
\label{sec:model}

A finite input $I$ consists of distinct item occurrences
$i=(s_i,a_i)$, where $s_i\in(0,1]$ is a size and $a_i\in\R_{\ge0}$ is an
arrival time.  Equal sizes and equal arrival times are allowed, but item
identities remain distinct.  All items at one epoch are revealed before any
action at that epoch.

At time $t$, an algorithm may perform a finite ordered trace of instantaneous
service actions.  Each action chooses a nonempty set $B$ of currently pending
items with $\sum_{i\in B}s_i\le1$, assigns it irrevocably to a fresh bin, and
seals that bin.  The complete arrival epoch is revealed before this trace
begins, and an item served by one action is no longer pending for later
actions in the same trace.  Each action costs
\begin{equation}
  \label{eq:bin-cost}
  c(B,t)=1+\max_{i\in B}(t-a_i).
\end{equation}
Every item must eventually be packed.  Although the input is finite, the
algorithm receives no end-of-input notification.  On the empty input we set
both online and offline costs to zero.

\begin{figure}[H]
  \centering
  \begin{tikzpicture}[
      x=1.18cm,y=0.9cm,
      arrival/.style={circle,draw=blue!65!black,fill=blue!10,inner sep=1.7pt},
      >=Latex]
    \draw[->,thick] (0,0) -- (7,0) node[right] {$t$};
    \foreach \x/\lab in {0.8/$a_1$,2.25/$a_2$,3.65/$a_3$} {
      \draw (\x,0.08)--(\x,-0.08);
      \node[arrival,above=3pt] at (\x,0) {};
      \node[below=5pt] at (\x,0) {\lab};
    }
    \draw[very thick,red!70!black] (5.55,-0.22)--(5.55,0.72);
    \node[above,red!70!black] at (5.55,0.72) {seal at $t_B$};
    \foreach \x in {0.8,2.25,3.65}
      \draw[densely dashed,gray!75,->] (\x,0.48)--(5.48,0.48);
    \node[draw,rounded corners,fill=green!8,align=center]
      at (3.5,-1.48)
      {$c(B,t_B)=\underbrace{1}_{\text{one setup}}
       +\underbrace{\max_{i\in B}(t_B-a_i)}_{\text{one oldest-item delay}}$};
  \end{tikzpicture}
  \caption{A bin is charged once for service and once for the delay of its
  oldest item.  The younger items do not contribute additional waiting
  charges.}
  \label{fig:model-cost}
\end{figure}
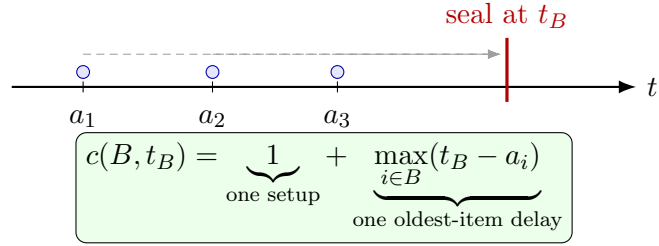

Let $\OPT(I)$ denote the minimum clairvoyant cost.  A deterministic algorithm is
strictly $\rho$-competitive if $\ALG(I)\le\rho\OPT(I)$ for every finite
input.  A randomized algorithm is strictly $\rho$-competitive against an
oblivious adversary if, for every fixed finite input $I$ independent of its
seed,
\[
  \E[\ALG(I)]\le \rho\OPT(I).
\]
Thus all randomized expectations below are over the algorithm's seed after
the input has been fixed.  We separate computationally unrestricted
algorithms, which may solve a strongly NP-hard packing problem at a service
epoch, from polynomial-time algorithms.  For algorithms that sample
continuous thresholds, polynomial time is in the standard ideal random-real
model: exact sampling, real arithmetic, and real comparisons are unit-cost
operations.  A bit-model implementation would require a separate
finite-precision or lazy-sampling statement.
Unless an ideal-real convention is explicitly invoked, every polynomial-time
claim assumes that item sizes, arrival times, ages, and requested accuracies
are binary-encoded rationals.  Fixed algebraic constants of bounded degree
are handled in their exact ordered fields when stated.

For a nonempty item set $B$, write
\[
  \spanof(B)=\max_{i\in B}a_i-\min_{i\in B}a_i,
  \qquad
  \vol(B)=\sum_{i\in B}s_i.
\]
For an ordinary bin-packing instance $S$, let $\BP(S)$ denote its minimum number
of unit-capacity bins.  The total input volume $V=\sum_i s_i$ satisfies
$V\le\BP(I)\le\OPT(I)$.

We use the following arrival-first rule throughout.  If an arrival epoch is
tied with a scheduled timer or grid boundary, reveal the complete epoch
first, process its items in a fixed identity order when an algorithm needs a
sequential scan, and only then execute any action still due.  Continuous
random thresholds make a fixed positive-gap tie a probability-zero event,
but the convention is required for deterministic statements and
simultaneous arrivals.

\section{Offline Temporal-Span Bin Packing}
\label{sec:offline}

The offline problem has a static partition formulation.

\begin{lemma}[Last-arrival normal form]
  \label{lem:last-arrival}
  There is an offline optimum that seals every bin at the last arrival of an
  item assigned to that bin.  Consequently,
  \begin{equation}
    \label{eq:offline-partition}
    \OPT(I)=\min_{\mathcal P}
       \sum_{B\in\mathcal P}\bigl(1+\spanof(B)\bigr),
  \end{equation}
  where the minimum ranges over capacity-feasible partitions of the items.
\end{lemma}

\begin{proof}
  In any feasible schedule, move the service time of a bin backward to the
  latest arrival assigned to it.  Every assigned item has arrived, capacity
  is unchanged, and the maximum delay weakly decreases.  Applying this to
  every bin proves the claim and~\eqref{eq:offline-partition}.
\end{proof}

We sort occurrences by nondecreasing arrival time, using the fixed identity
order for ties, and write $g_j=a_{j+1}-a_j$.  For a partition $\mathcal P$,
let $N_j(\mathcal P)$ denote the number of bins containing an item on each side
of cut $j$.

\begin{lemma}[Cut congestion]
  \label{lem:cut-congestion}
  Every feasible partition satisfies
  \begin{equation}
    \label{eq:cut-congestion}
    \sum_{B\in\mathcal P}(1+\spanof(B))
    =|\mathcal P|+\sum_{j=1}^{n-1}g_jN_j(\mathcal P).
  \end{equation}
\end{lemma}

\begin{proof}
  The span of a bin is the sum of exactly those consecutive arrival gaps
  crossed by its temporal hull.  Interchanging the two finite sums gives
  \eqref{eq:cut-congestion}.
\end{proof}

Splitting every bin at a cut increases the setup term by at most one per
crossing bin and removes $g_j$ per crossing bin.  Hence an optimum can be
chosen not to cross a gap $g_j\ge1$, and no optimum crosses a gap $g_j>1$.
This decomposition is sharp at equality.
Appendix~\ref{app:offline-phase} gives the split calculation.

Simultaneous arrivals recover ordinary bin packing, giving strong
NP-hardness and the absolute $3/2$ approximation barrier.
Section~\ref{sec:offline-derand} gives a polynomial $2$-approximation.
Equal sizes, nonbinding capacity, and bin cardinality at most two admit
exact polynomial algorithms. Appendix~\ref{app:offline-computation}
states these consequences and links their complete proofs.

For nonbinding capacity, the same cut identity gives the formula used
below: $\OPT(I)=1+\sum_{j=1}^{n-1}\min\{g_j,1\}$.

\section{Random Silence with Binding Capacity}
\label{sec:randomized-event}

The algorithm combines a silence timer with one Next-Fit buffer. We first
specify its service rule, then compare each offline bin with the pieces
created by silence. The final step keeps the two-unit saving per cluster
when adding the temporal cost. Before arrivals begin, sample a global
threshold uniformly on $[0,1]$, so that
\begin{equation}
  \label{eq:silence-density}
  \Pr[X\le x]=x \qquad(0\le x\le1).
\end{equation}

\begin{algorithm}[H]
  \caption{Capacity-Overflow Random Silence}
  \label{alg:cors}
  \begin{algorithmic}[1]
    \State sample one $X$ according to~\eqref{eq:silence-density}
    \State $R\gets\varnothing$ and timer $\tau\gets+\infty$
    \While{the online process is running}
      \State wait for the next arrival epoch $t$ or the timer $\tau$
      \If{an arrival epoch occurs at $t\le\tau$}
        \State reveal the full epoch and scan it in fixed identity order
        \For{each item $i$ in this epoch}
          \If{$R\ne\varnothing$ and $\vol(R)+s_i>1$}
            \State seal $R$ at time $t$ and set $R\gets\varnothing$
          \EndIf
          \State $R\gets R\cup\{i\}$ and reset $\tau\gets t+X$
        \EndFor
      \Else
        \State seal $R$ at time $\tau$; set $R\gets\varnothing$ and $\tau\gets+\infty$
      \EndIf
    \EndWhile
  \end{algorithmic}
\end{algorithm}

The same $X$ is retained throughout the input.  At a tied epoch, the complete
epoch is revealed first, capacity overflows are executed during its fixed-order
scan, and only after that scan does a timer that remains due fire.  Thus a
boundary arrival resets the timer before any still-due silence action.  The
algorithm is polynomial in the ideal random-real model, uses one buffer, and
the final timer provides service without an end signal.

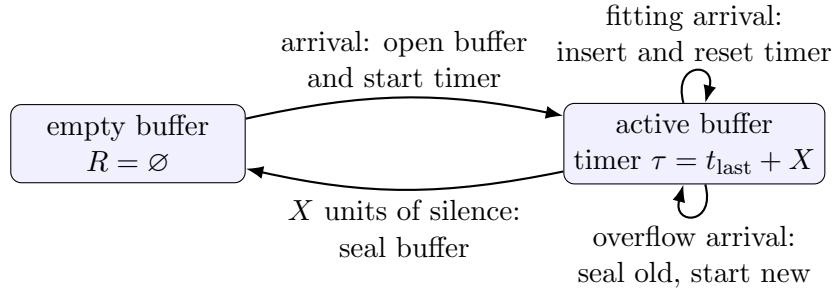
\begin{figure}[H]
  \centering
  \begin{tikzpicture}[
      >=Latex,
      state/.style={draw,rounded corners,minimum width=31mm,minimum height=10mm,
        align=center,fill=blue!6},
      every edge/.style={draw,->,thick}]
    \node[state] (empty) {empty buffer\\$R=\varnothing$};
    \node[state,right=42mm of empty] (active) {active buffer\\timer $\tau=t_{\rm last}+X$};
    \path (empty) edge[bend left=12]
      node[above,align=center] {arrival: open buffer\\and start timer} (active);
    \path (active) edge[loop above,looseness=5]
      node[above,align=center] {fitting arrival:\\insert and reset timer} (active);
    \path (active) edge[loop below,looseness=5]
      node[below,align=center] {overflow arrival:\\seal old, start new} (active);
    \path (active) edge[bend left=12]
      node[below,align=center] {$X$ units of silence:\\seal buffer} (empty);
  \end{tikzpicture}
  \caption{The event-driven state machine.  Capacity overflow changes the
  bin but not the arrival-driven silence process; a silence event empties
  the system until the next arrival.}
  \label{fig:cors-state}
\end{figure}

\subsection{A pathwise capacity decomposition}

Let $C_X(I)$ be the cost of the same silence rule with capacity ignored.
For fixed $X$, we call a bin an \emph{overflow bin} if it is closed by a
nonfitting arrival, and let $F_X$ denote their number.  Every consecutive gap is
either internal to a buffer, paid when an overflow closes the preceding
buffer at the later arrival, or replaced by the terminal wait $X$ after a
silence cut.  Capacity therefore adds openings but no temporal term:
\begin{equation}
  \label{eq:pathwise-decomposition}
  \ALG_X(I)=F_X+C_X(I).
\end{equation}
Indeed, we fix one nonempty silence cluster $J$. If Next Fit creates buffers
$K_1,\ldots,K_m$, let $b_r$ denote the first arrival in $K_r$ and let $z$ be
the last arrival in $J$.  The first $m-1$ buffers close at the next buffer's
first arrival and the last closes at $z+X$, and therefore
\begin{align*}
 \ALG_X(J)&=\sum_{r<m}(1+b_{r+1}-b_r)+(1+z+X-b_m)\\
 &=(m-1)+(1+z+X-b_1).
\end{align*}
The parenthesized term is exactly $C_X(J)$ and $m-1$ is the number of
overflows in $J$.  Summing proves~\eqref{eq:pathwise-decomposition}, including
zero gaps and multiple overflows at one epoch.  On the empty input set
$F_X=S=H_X=0$, and therefore the identity is trivial.

Let $J_1,\ldots,J_S$ denote the global $X$-silence clusters.  Within $J_c$, the
algorithm is ordinary Next Fit.  If it uses $m_c$ bins, then
\begin{equation}
  \label{eq:nf-absolute}
  m_c\le2\BP(J_c)-1.
\end{equation}
Indeed, consecutive Next-Fit bins have combined load greater than one, and
pairing bins proves the integer bound.  Since $F_X=\sum_c(m_c-1)$,
\begin{equation}
  \label{eq:overflow-clusters}
  F_X\le2\sum_c\BP(J_c)-2S.
\end{equation}

We fix an offline optimum $\mathcal O$.  Let $h_B(X)$ denote the number of silence
clusters meeting offline bin $B$, and put $H_X=\sum_{B\in\mathcal O}h_B(X)$.
Restricting each offline bin to each cluster gives
\begin{equation}
  \label{eq:cluster-fragments}
  \sum_c\BP(J_c)\le H_X.
\end{equation}

\begin{figure}[H]
  \centering
  \begin{tikzpicture}[x=0.82cm,y=0.75cm,>=Latex]
    \draw[->,thick] (0,0)--(11,0) node[right] {$t$};
    \foreach \x in {0.8,1.7,2.7,3.8,5.1,6.2,7.5,8.7,9.8}
      \fill[black] (\x,0) circle (1.6pt);
    \draw[densely dashed,red!70!black] (3.25,-2.55)--(3.25,0.65);
    \draw[densely dashed,red!70!black] (6.85,-2.55)--(6.85,0.65);
    \node[above] at (1.65,0.3) {$J_1$};
    \node[above] at (5.05,0.3) {$J_2$};
    \node[above] at (8.85,0.3) {$J_3$};
    \node[left] at (0,-1.05) {offline bin $B$};
    \draw[very thick,blue!65!black] (1.7,-1.05)--(7.5,-1.05);
    \foreach \x in {1.7,5.1,7.5}
      \fill[blue!65!black] (\x,-1.05) circle (2.2pt);
    \node[right,blue!65!black] at (7.7,-1.05) {$h_B(X)=3$};
    \node[left] at (0,-2.05) {offline bin $B'$};
    \draw[very thick,green!50!black] (2.7,-2.05)--(6.2,-2.05);
    \foreach \x in {2.7,3.8,6.2}
      \fill[green!50!black] (\x,-2.05) circle (2.2pt);
    \node[right,green!50!black] at (6.4,-2.05) {$h_{B'}(X)=2$};
  \end{tikzpicture}
  \caption{Silence cuts fragment a fixed offline bin only across gaps in
  its temporal hull.  Counting the resulting colored fragments yields
  $H_X$, which simultaneously upper-bounds the bin-packing demand inside
  all silence clusters.}
  \label{fig:offline-fragments}
\end{figure}
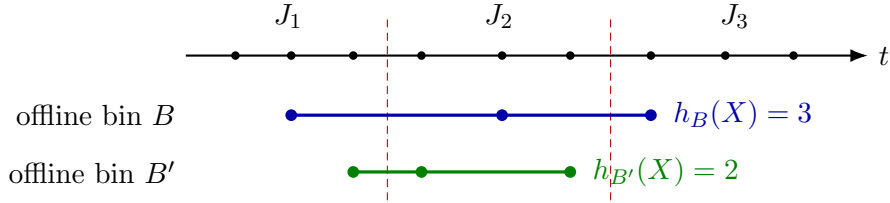
A bin can enter a new silence cluster only when a global consecutive gap is
cut between its first and last labelled occurrence in the global
arrival/identity order.  More explicitly, globally order the labelled
occurrences as $i_1,\ldots,i_n$, put
$d_j=a_{i_{j+1}}-a_{i_j}$, and let $r,s$ denote the first and last global indices
occupied by $B$.  Pathwise,
\[
 h_B(X)\le 1+\sum_{j=r}^{s-1}\mathbf 1\{X<d_j\},
 \qquad \sum_{j=r}^{s-1}d_j=\spanof(B).
\]
The inequality, rather than equality, allows clusters containing only items
outside $B$.  Since
\[
  \Pr[X<d]= d\quad(0\le d\le1)
\]
and the bound is trivial for $d\ge1$, the hull gaps telescope to give
\begin{equation}
  \label{eq:expected-fragments}
  \E[h_B(X)]\le1+\spanof(B),
  \qquad
  \E[H_X]\le\OPT(I).
\end{equation}

\subsection{The signed silence potential}

The negative term $-2S$ in~\eqref{eq:overflow-clusters} is the source of the
improved constant.  Gap by gap,
\begin{align}
  C_X&=(1+X)+\sum_d
    \bigl(d\mathbf 1[X\ge d]+(1+X)\mathbf 1[X<d]\bigr),\\
  S&=1+\sum_d\mathbf 1[X<d].
\end{align}
The initial contribution to $\E[C_X-2S]$ is
$\E[X-1]=-1/2$. For $0\le d\le1$, an uncut gap costs $d$; a cut
contributes $X-1$ after subtracting the cluster saving. Its mean is
\begin{equation}
  \label{eq:signed-gap}
  h(d)=d(1-d)+\int_0^d(x-1)\,dx=-d^2/2.
\end{equation}
For $d\ge1$ the contribution is $-1/2$. On every nonempty input,
\begin{equation}
  \label{eq:signed-potential}
  \E[C_X-2S]=-\tfrac12\bigl(1+\sum_j\min\{g_j,1\}^2\bigr).
\end{equation}

\begin{theorem}[Main randomized upper bound]
  \label{thm:main-randomized}
  Against every oblivious adversary, Capacity-Overflow Random Silence
  has strict competitive ratio $2$ in expectation. This factor is tight
  for the stated fixed-order Next-Fit policy.
\end{theorem}

\begin{proof}
  Combining~\eqref{eq:pathwise-decomposition},
  \eqref{eq:overflow-clusters}, and~\eqref{eq:cluster-fragments} gives
  \[
    \ALG_X\le2H_X+(C_X-2S).
  \]
  We take expectations and apply~\eqref{eq:expected-fragments} and
  \eqref{eq:signed-potential}. For tightness, at a single epoch present
  $m$ alternating pairs of sizes $1,1/m$ in identity order, with $m\ge2$.
  Every pair of neighbors overflows, so the policy pays $2m+X$.
  OPT packs all $m$ small items together and uses $m$ singleton bins,
  at cost $m+1$. The ratio tends to $2$.
\end{proof}

\begin{remark}
  The exponential threshold attains the optimal $e/(e-1)$ ratio when
  capacity is relaxed (Appendix~\ref{sec:lower-bounds}). With binding
  capacity, a uniform threshold makes the expected signed cost of every
  gap nonpositive. Appendix~\ref{app:finite-silence} gives a
  $b$-random-bit implementation with strict ratio at most $2+2^{-b}$.
\end{remark}

\section{A deterministic silence threshold}\label{sec:fixed-silence}
The same buffer and cut comparison give a deterministic algorithm. The
threshold now balances the cost of gaps between disjoint offline bins
against the fragmentation of bins that span a silence cut.

\begin{theorem}\label{thm:fixed-silence}
With the fixed threshold $\tau=\sqrt3-1$, Algorithm~\ref{alg:cors}
is strictly $(1+\sqrt3)$-competitive in polynomial time.
\end{theorem}
\begin{proof}
Fix $0<\tau\le1$ and an optimal partition with $k$ bins, and write $q$
for the number of offline hulls crossing a consecutive gap $d$. By
Lemma~\ref{lem:cut-congestion}, $\OPT=k+\sum_d qd$. As before
$\ALG\le2H_\tau+(C_\tau-2S)$, with
$H_\tau\le k+\sum_{d>\tau}q$. The initial signed cost $\tau-1$ is
nonpositive. An uncut gap contributes $d$ to the signed cost and a
cut gap contributes $\tau-1$.

When $q\ge1$, the total gap charge is $d$ for $d\le\tau$, and at most
$2q$ for $d>\tau$; in either case it is at most $(2/\tau)qd$.
When $q=0$, the charge is at most $\tau$. There are at most $k-1$
such gaps: cutting the sequence at all of them gives nonempty blocks
with disjoint sets of offline bins. Consequently
\[
 \ALG\le(2+\tau)k+(2/\tau)\sum_d qd.
\]
Taking $\tau=\sqrt3-1$ makes both coefficients $1+\sqrt3$.
For rational gaps the test $d>\tau$ is equivalent to $(1+d)^2>3$;
timer comparisons and service times use the fixed field
$\mathbb Q(\sqrt3)$.
\end{proof}

\section{The Exact Temporal Baseline and General Lower Bounds}
\label{sec:lower-bounds}

This section isolates the case in which all items together fit in one bin.
It supplies lower bounds inherited by the full model and the optimal
exponential threshold for the capacity-relaxed objective. The uniform
threshold in Section~\ref{sec:randomized-event} instead accounts
for the additional cost of capacity overflows.

\paragraph{Causal policies.}\label{par:causal-policies}
The pathwise transformations in this section and the stochastic analysis
use the following common measurability convention.
Formally, the master private seed $U$ is sampled from a standard Borel
probability space independently of the input.  Finite labelled histories and
finite action traces carry their natural standard Borel structures.  A
randomized causal policy is jointly measurable in $(U,h)$, and after $U=u$
is fixed it is a deterministic causal policy whose service times are stopping
times for the arrival history.  We require pathwise admissibility for every
seed under consideration; changing a policy on a null seed set to a fixed
safe fallback is harmless.  Pathwise policy transformations below are
applied seed by seed, and expectations are taken only after the transformed
executions and their joint measurability have been established.

When the total size is at most one, the following normalization permits us to
reason about temporal batches rather than arbitrary retained subsets.

\begin{lemma}[Oldest-trigger normalization]
\label{lem:one-bin-normalization}
For every online policy $A$, and pathwise for every fixed seed when $A$ is
randomized, there is a causal policy $A'$ that clears all pending items
whenever it acts and satisfies $A'(I)\le A(I)$ on every input whose total
size is at most one.
\end{lemma}

\begin{proof}
  First extend the construction to a policy on the full input domain.  If its
  pending volume ever exceeds one, abandon the shadow transformation, pack
  all currently pending items feasibly (singletons suffice), and thereafter
  use the singleton policy.  This fallback is causal and is never reached on
  the promised class.

  We shadow-simulate $A$.  Let $q$ denote the oldest item pending for $A'$.  When the
shadow policy first packs $q$, at time $t$, policy $A'$ packs every item then
pending.  This is feasible by the total-size assumption, and its cost
$1+t-a_q$ is no greater than the shadow bin containing $q$.  Different
actions of $A'$ charge different shadow bins.  Items served earlier in the
shadow remain physically pending until their oldest trigger, and therefore the
construction is causal and eventually serves every item whenever $A$ does.
Summing the injective charges proves the claim.  The first-exceed fallback
event and the first shadow service containing the current oldest labelled
item are measurable stopping-time selections.  Induction over the finitely
many identities therefore makes $A'$ jointly measurable in $(U,h)$ under
the policy convention above.
\end{proof}

The exact offline value is~\eqref{eq:one-bin-opt}.

\begin{theorem}
  \label{thm:det-lb}
  Every deterministic algorithm has competitive ratio at least two, even
  when all items together fit in one bin.  Starting a nonresetting unit timer
  whenever the system becomes nonempty and clearing all pending items when it
  fires is $2$-competitive on this class.
\end{theorem}

The adversarial lower-bound construction and its algebra are given in
Appendix~\ref{app:lower-bounds}.  For the upper bound, when the system becomes
nonempty at a trigger $r_j$, start a nonresetting unit timer.  The algorithm
pays exactly two per action.  Arrival-first processing implies
$r_{j+1}>r_j+1$.  If one offline bin contains $h$ triggers, its span is
greater than $h-1$ and its cost is at least $h$.  Assigning every trigger to
its offline bin gives $\OPT\ge M$ and $\ALG=2M\le2\OPT$.  The last ordinary
timer fires without an end signal.

For an oblivious randomized lower bound, fix integers $m,K\ge2$, put
$\delta=1/m$, and use items of size $1/K$.  After each generated item, choose
the next gap $j/m$ with probability
\[
  p_j=\delta(1-\delta)^{j-1},\qquad j=1,\ldots,m,
\]
or terminate with probability $(1-\delta)^m$; force termination after item
$K$.  This is a finite-support distribution.  Appendix~\ref{app:lower-bounds}
proves the equalizer inequality for every deterministic response, computes
the exact expected offline denominator, takes the limits in the order
$K\to\infty$ and then $m\to\infty$, and performs the finite-support averaging
over a randomized policy's seed.

\begin{theorem}
  \label{thm:rand-lb}
  Every randomized algorithm has competitive ratio at least $e/(e-1)$
  against an oblivious adversary, even when the total size is at most one.
\end{theorem}

The matching upper bound samples one global $\Theta\in[0,1]$ with density
$e^\theta/(e-1)$, resets a silence timer after every arrival, and clears all
pending items after $\Theta$ units of silence.  If $P_\theta$ is the resulting
block partition, $m(\theta)=|P_\theta|$, and
\[
  Q(\theta)=m(\theta)\theta+
    \sum_{j:g_j\le\theta}g_j,
\]
then $\ALG_\theta=m(\theta)+Q(\theta)$ and
$Q'(\theta)=m(\theta)$ away from gap breakpoints.  The function $Q$ is
continuous, $Q(0)=0$, and $Q(1)=\OPT$.  Hence
\[
  \int_0^1e^\theta\ALG_\theta\,d\theta
  =\int_0^1(e^\theta Q(\theta))'\,d\theta
  =e\OPT.
\]
This proves that $e/(e-1)$ is the exact randomized ratio on the nonbinding
capacity class.  The same $\Theta$ is used globally and the final silence
timer supplies cleanup without an end signal.

\section{Beyond the Efficient Online Bound}
\label{sec:further-overview}

The main algorithm handles capacity without solving an optimization
problem at service time. Two complementary comparisons help locate its
guarantee: exact packing at prescribed times, and exact analysis of a
stochastic input with binding capacity. Table~\ref{tab:frontier} collects
these results and separates their performance measures.

\begin{table}[!htbp]
  \centering
  \small
  \begin{tabularx}{\textwidth}{@{}
    >{\raggedright\arraybackslash}p{0.29\textwidth}
    X
    >{\raggedright\arraybackslash}p{0.24\textwidth}@{}}
    \toprule
    \multicolumn{3}{@{}l}{\textbf{A. Offline approximation and computation}} \\
    Instance class & Algorithmic result & Offline guarantee \\
    \midrule
    General & Hardness and absolute inapproximability & no ratio $<3/2$ \\
    General & Polynomial approximation & $2$ \\
    Equal sizes, nonbinding, or cardinality at most two
      & Polynomial exact algorithm & exact \\
    Fixed-length weighted endpoint packing & AFPTAS
      & $(1+\varepsilon)\OPT+\operatorname{poly}_L(1/\varepsilon)$ \\
    \addlinespace
    \midrule
    \multicolumn{3}{@{}l}{\textbf{B. Adversarial online competitiveness}} \\
    Model and algorithm class & Competitive comparison & Strict ratio \\
    \midrule
    Capacity-constrained, deterministic & lower / unrestricted upper
      & $2$ / $(3+\sqrt5)/2$ \\
    Capacity-constrained, deterministic polynomial & lower / event-driven upper
      & $2$ / $1+\sqrt3$ \\
    Capacity-constrained, randomized & lower / unrestricted upper
      & $e/(e-1)$ / $1+1/\sqrt2$ \\
    Capacity-constrained, randomized ideal-real polynomial & lower / CORS upper
      & $e/(e-1)$ / $2$ \\
    Nonbinding capacity & Exact deterministic / randomized frontiers
      & $2$ / $e/(e-1)$ \\
    \addlinespace
    \midrule
    \multicolumn{3}{@{}l}{\textbf{C. Poisson half-size stochastic benchmark}} \\
    Quantity & Status & Long-run value \\
    \midrule
    Offline cost rate & Exact
      & $\lambda(\lambda+2)/(2(\lambda+1))$ \\
    Optimal causal online rate & Exact
      & $\lambda$ below rate one; $(\lambda+1)/2$ above \\
    Ratio of expectations and conditional expected ratio
      & Common asymptotic limit & at most $4/3$ \\
    \bottomrule
  \end{tabularx}
  \caption{Results grouped by performance metric.  Panel A concerns offline
  approximation, Panel B strict adversarial competitive ratios, and Panel C
  long-run stochastic rates.  Randomized adversarial guarantees are against
  an oblivious adversary.}
  \label{tab:frontier}
\end{table}

\paragraph{Exact packing at grid endpoints.}
Fix a phase length $L$ and dispatch all arrivals in each phase at its
right endpoint. The remaining problem is weighted endpoint packing:
each bin costs one plus the age of its oldest item at that endpoint.
Splitting an offline bin at phase boundaries pays for the additional
pieces using the bin's original arrival span. A deterministic grid and
a uniformly shifted grid give different boundary-crossing charges;
optimizing $L$ yields ratios $(3+\sqrt5)/2$ and $1+1/\sqrt2$,
respectively (Appendix~\ref{sec:phase}). The random shift is sampled
independently of the whole input.

These ratios describe what is possible with exact endpoint packing.
Computing that packing can be strongly NP-hard. The distinction matters:
the randomized oracle ratio is smaller than the efficient CORS bound,
but the oracle algorithm is not thereby a polynomial-time improvement.
Our endpoint AFPTAS approximates this computational subproblem.

\paragraph{Why the endpoint problem needs age information.}
An ordinary packing approximation controls the number of bins. At a
common endpoint, it must also control how many bins contain old items.
We round ages and large sizes and use configurations that retain age
classes. Small items then face nested residual-capacity constraints:
older items can use only bins whose charged age is at least their own.
Respecting these nested constraints is the additional step beyond
ordinary configuration-LP packing. Knapsack pricing and integral
reconstruction give, for fixed rational $L$, a guarantee
$(1+\varepsilon)\OPT+\operatorname{poly}_L(1/\varepsilon)$ for one
endpoint instance (Theorem~\ref{thm:weighted-afptas}). Its additive term
would be paid in every phase. Consequently this theorem does not claim
an AFPTAS for the unrestricted temporal problem or the exact oracle
ratios for an efficient online algorithm.

\paragraph{An exact benchmark with capacity already active.}
Half-size items can share a bin in pairs, so Poisson arrivals reduce the
packing choice to a temporal matching problem. Offline, crossing pairs
can be uncrossed and an optimum pairs adjacent arrivals. The marginal
saving in the resulting dynamic program has a stationary distribution,
which yields the exact offline cost rate. Online, a causal exchange
permits immediate service whenever two items are pending. Only the
timeout of a single waiting item remains to be chosen.

At arrival rate $\lambda\le1$, immediate single-item service is optimal.
For $\lambda>1$, increasing finite timeouts attain the optimal rate while
still clearing a final unmatched item without an end signal. A
conditional cycle-cost inequality gives the matching bound for every
admissible causal policy, not only fixed timers. Appendix~\ref{sec:stochastic-half}
also justifies the limiting ratios: the ratio of expectations and the
expected realized ratio, conditional on nonempty input, have the same
limit, at most $4/3$. This stochastic result complements the worst-case competitive guarantees.

\section{Conclusion and Open Questions}
\label{sec:conclusion}

The per-bin maximum-delay objective gives an offline partition cost equal
to the number of bins plus their total arrival span. For online packing,
CORS compares the fragments of this offline packing with silence clusters
and retains the saving of two openings per cluster. This proves its
strict $2$ ratio in the ideal random-real model.

\paragraph{Competitive ratios.}
With unrestricted computation, the randomized optimum lies between
$e/(e-1)$ and $1+1/\sqrt2$; for polynomial algorithms in the stated
ideal-real model, our upper bound is $2$.
The deterministic bounds are $2$ and $(3+\sqrt5)/2$ with unrestricted
computation, and $2$ and $1+\sqrt3$ in polynomial time. A lower bound above the
nonbinding constants, or algorithms attaining those constants with
binding capacity, would determine whether capacity changes the optimal
competitive ratios themselves.

\paragraph{Offline approximation.}
The full temporal-span problem has neither a sharp known approximation
threshold nor an AFPTAS. The fixed-length endpoint AFPTAS incurs its
additive term separately in every occupied phase. Controlling this repeated
loss, or improving the full problem's factor-$2$ approximation,
would strengthen the computational guarantees. An absolute PTAS is ruled
out by the $3/2$ barrier unless $\mathrm P=\mathrm{NP}$; this does not
rule out an asymptotic scheme.

\paragraph{Stochastic inputs.}
For Poisson half-size arrivals, matching and causal policy normalization
determine both cost rates and a common limiting ratio of at most $4/3$.
The alphabet $\{1/3,2/3\}$ is a further case: the choice between
complementary pairs and triples prevents the immediate-pair normalization
used here. This alphabet and continuous size distributions remain open directions.

\clearpage
\appendix
\section{Finite randomness for the silence rule}\label{app:finite-silence}
The proof of the continuous-threshold theorem is in Section~\ref{sec:randomized-event}.
For rational input, we can approach its guarantee with a prescribed number of
fair bits and rational comparisons.

\begin{theorem}\label{thm:finite-cors}
For $b\ge0$, sample the global threshold uniformly from
$\{1/M,2/M,\ldots,1\}$, where $M=2^b$. This uses $b$ fair bits and is
strictly $(2+1/M)$-competitive against an oblivious adversary.
\end{theorem}

\begin{proof}
Fix a nonempty input and an optimal partition with $k$ bins. For each
consecutive gap $d$, let $q$ count the offline bin hulls that cross it.
Thus $\OPT=k+\sum_d qd$ by Lemma~\ref{lem:cut-congestion}. With
$F(d)=\Pr[X<d]$, the fragment inequality gives
$\E[H_X]\le k+\sum_d qF(d)$. The initial term of $\E[C_X-2S]$ is
$\E[X-1]\le0$. Write $h(d)$ for the remaining signed contribution
from this gap.

When $0\le d\le1$, put $u=F(d)$ and $z=d-u$, so $0\le z\le1/M$,
including ties under the arrival-first rule. The possible thresholds
strictly below $d$ are $1/M,\ldots,u$, each of probability $1/M$. Hence
\[
 h(d)=d(1-u)-u+\frac{u(u+1/M)}2
     =z(1-u)+\frac{u(1/M-u)}2\le z.
\]
The last inequality holds since $u=0$ or $u\ge1/M$. If $d>1$, then
$F(d)=1$ and $h(d)=\E[X-1]\le0$. In particular, a covered gap ($q\ge1$)
satisfies $2qF(d)+h(d)\le2qd$. An uncovered gap has $h(d)\le1/M$.
At most $k-1$ gaps are uncovered: splitting at all such gaps partitions
the nonempty sequence into blocks with disjoint nonempty sets of offline
bins. The pathwise bound from Theorem~\ref{thm:main-randomized}
therefore gives
\[
 \E[\ALG]\le2\OPT+(k-1)/M\le(2+1/M)\OPT.
\]
The empty input costs zero. Sampling a uniform $b$-bit integer, adding
one, and dividing by $M$ gives the stated threshold exactly.
\end{proof}

\section{A Deterministic Polynomial Baseline}
\label{sec:deterministic}

The following polynomial algorithm maintains one pending buffer $R$ of load
strictly below $1/2$.  An arriving item of size at least $1/2$ is sealed
immediately, together with $R$ if they fit.  A smaller item is inserted into
$R$; when the load reaches $1/2$, the whole buffer is sealed.  Finally, a
nonempty buffer is sealed when its oldest item reaches age one.  Arrival
actions precede a tied timer.

\begin{theorem}
  \label{thm:half-full}
  Half-Full-or-One is strictly $4$-competitive.
\end{theorem}

\begin{proof}
  Let $F$ denote the number of arrival-triggered bins, $K$ the number of timer
  bins, $D$ the sum of all algorithmic maximum delays, and $V$ the total
  volume.  Every arrival-triggered bin has load at least $1/2$, and therefore
  $F\le2V\le2\OPT$.  Every timer bin has delay one, hence $K\le D$.

  We list positive-delay services in chronological order.  If $u_j$ is the
  arrival of the oldest item in service $j$, $t_j$ its service time, and
  $d_j=t_j-u_j$, then $0<d_j\le1$.  Since every positive-delay service clears
  the entire old buffer, $u_{j+1}\ge t_j=u_j+d_j$.  We retain one oldest witness
  per service and relax capacity.  Formula~\eqref{eq:one-bin-opt} for the
  witness arrivals gives
  \[
    \OPT\ge1+\sum_j\min\{u_{j+1}-u_j,1\}\ge\sum_jd_j=D.
  \]
  Therefore $\ALG=F+K+D\le2V+2D\le4\OPT$, which proves the theorem.
\end{proof}

The analysis is not tight for this algorithm.  Three compatible small items
at times $0,\varepsilon,2\varepsilon$, with the first two crossing load
$1/2$, give ratio tending to three.  Exact searches have not found a larger
ratio, but no factor-three proof is known.

\section{Periodic Phase Decompositions}
\label{sec:phase}

CORS gives ratio $2$ in the ideal random-real model. A smaller randomized
upper bound is possible if service epochs may solve a
strongly NP-hard packing problem.  Periodic decomposition exposes this
computational tradeoff: exact endpoint optimization gives stronger
unrestricted ratios, and the local oracle motivates the weighted approximation in Appendix~\ref{sec:weighted-afptas}.

We fix a phase length $L>0$.  For grid offset $U$, the phases are the right-closed
intervals $(U+(k-1)L,U+kL]$, $k\in\mathbb Z$.  An item at a boundary belongs
to the phase ending there: the entire boundary epoch is revealed before that
phase is served, consistently with the arrival-first convention.  At the
right endpoint $r$ of a phase, let $S$ be
the items arriving in that phase.  The exact phase algorithm chooses a
capacity-feasible partition minimizing
\begin{equation}
  \label{eq:phase-objective}
  q_r(\mathcal P)=\sum_{C\in\mathcal P}
    \left(1+r-\min_{i\in C}a_i\right)
\end{equation}
and seals all its bins at $r$.  This problem is strongly NP-hard: placing all
items at the right endpoint makes every age zero and recovers ordinary bin
packing.

Restrict every bin of a fixed offline optimum to its nonempty phase
fragments.  Those fragments form a feasible competitor for every phase solve.
If the grid is fixed, a direct endpoint calculation gives
\begin{equation}
  \label{eq:det-phase}
  R_{\mathrm{det}}(L)=\max\{2+L,1+1/L\}.
\end{equation}
The two branches balance at $L=(\sqrt5-1)/2$, giving
$(3+\sqrt5)/2$.

For completeness, if one offline bin of span $\Delta$ meets $m\ge2$ phases
and $x\in[0,L)$ is the distance from its first item to its phase's right
boundary,
its fragment cost is at most
\[
 1+x+(m-1)(1+L),\qquad
 \Delta\ge x+(m-2)L.
\]
The quotient is maximized either by two fragments with $x\downarrow0$,
giving $2+L$, or by many crossed phases, giving $1+1/L$.  Both branches are
attained as suprema over finite inputs.  For the first, release two half-size
items at $L-\varepsilon$ and $L+\varepsilon$.  Exact phase optimization pays
$2+L$, while the offline optimum pays $1+2\varepsilon$; the ratio tends to
$2+L$.  For the second, distribute items of total volume one over a fine mesh
spanning $m$ consecutive phases, avoiding every boundary by $\varepsilon$.
The phase
algorithm uses one bin of cost $1+L-\varepsilon$ per phase, whereas the
nonbinding offline formula gives $1+mL-2\varepsilon$.  Letting first
$m\to\infty$ and then $\varepsilon\downarrow0$ gives $1+1/L$.  This proves the exact ratio in~\eqref{eq:det-phase}.

If the grid has a uniform random offset in $[0,L)$, independently of the
fixed input, an offline bin of span $\Delta$ has expected fragment cost at
most
\[
  1+\frac L2+\left(1+\frac1L\right)\Delta.
\]
Therefore
\begin{equation}
  \label{eq:random-phase}
  R_{\mathrm{rand}}(L)=\max\{1+L/2,1+1/L\},
\end{equation}
which is minimized at $L=\sqrt2$ and equals $1+1/\sqrt2$.

To see the coefficient of $\Delta$, expose one consecutive gap $g$.  Its
expected incremental fragment cost is
\[
 h_L(g)=
 \begin{cases}
 g/L+g-g^2/(2L),&0\le g<L,\\
 1+L/2,&g\ge L,
 \end{cases}
 \qquad h_L(g)\le(1+1/L)g.
\]
The first fragment costs $1+L/2$ in expectation.  Summing over the gaps in
one fixed offline bin proves the displayed estimate.  The offset is one
global seed, independent of the fixed input.

Both branches in~\eqref{eq:random-phase} are necessary for this algorithm.
A single item has uniform residual time to the next boundary and gives
$1+L/2$ exactly.  For the other branch, put $n+1$ items of total volume one
at $0,h,\ldots,nh$, where $0<h<\min\{1,L\}$.  The exact gap calculation above
gives
\[
 \E[\ALG]=1+\frac L2+n\left(\frac hL+h-\frac{h^2}{2L}\right),
 \qquad \OPT=1+nh.
\]
Letting $n\to\infty$ and then $h\downarrow0$ gives $1+1/L$.  Hence
\eqref{eq:random-phase} is again the exact supremal ratio against an
oblivious adversary.

More generally, if a positive global random $L$ is followed by a
conditionally uniform offset, with $\E[L]<\infty$ and $\E[1/L]<\infty$, the exact
ratio within that family is
\[
  \max\{1+\E[L]/2,1+\E[1/L]\}\ge1+1/\sqrt2.
\]
The dense-mesh limit is justified by dominated convergence with dominator
$1+1/L$.

For the renewal comparison, let the boundaries form an equilibrium
stationary renewal process on $\mathbb R$, sampled independently of the
input.  Under its Palm distribution let the positive i.i.d. inter-boundary
lengths have law $Y$, with $0<\mu=\E[Y]<\infty$ and
$\E[Y^2]<\infty$.  A single item is a finite witness whose expected cost is
$1+\E[Y^2]/(2\mu)$, by the equilibrium forward-recurrence formula, while its
offline cost is one.

For the second finite family, place items of total volume one on a mesh of
width $0<h<1$ in $[0,H]$.  The offline cost divided by $H$ tends to one as
$H\to\infty$.  For fixed $h$, the renewal-reward theorem gives limiting
setup and endpoint-delay rates $s(h)$ and $d(h)$ for the occupied renewal
intervals.  Palm--Campbell counting and dominated convergence give
$s(h)\to1/\mu$ and $d(h)\to1$ as $h\downarrow0$.  The two boundary intervals
contribute $O(1)$ in expectation, because the equilibrium forward and
backward recurrence times have finite means under $\E[Y^2]<\infty$.
Consequently
\[
 \lim_{h\downarrow0}\lim_{H\to\infty}
 \frac{\E[\ALG_{H,h}]}{H}=1+\frac1\mu,
 \qquad
 \lim_{H\to\infty}\frac{\OPT_{H,h}}H=1.
\]
These two finite limiting families give the lower barrier
\[
  \max\left\{1+\frac{\E[Y^2]}{2\E[Y]},
              1+\frac1{\E[Y]}\right\}
  \ge1+\frac1{\sqrt2}.
\]
The inequality uses $\E[Y^2]\ge\mu^2$ and
$\max\{\mu/2,1/\mu\}\ge1/\sqrt2$; equality forces
$Y=\sqrt2$ almost surely.  Thus neither the global-random-period family nor
the stationary-renewal family improves the uniform shifted grid.  The
renewal display is a universal lower barrier, not a claim that it is the
exact ratio of every renewal law.  Arbitrary nonstationary or nonrenewal
input-independent point processes are not covered.

\subsection{Polynomial weighted phase packing}

At endpoint $r$, give item $i$ age $w_i=r-a_i$ and process items in
nonincreasing age order with a prefix-consistent ordinary online bin-packing
algorithm $A$.  Suppose every finite prefix $Q$ satisfies
$A(Q)\le c\BP(Q)$.  For the final packing $\mathcal P$, let
$S_x=\{i:w_i\ge x\}$ and let $N_{\mathcal P}(x)$ count bins meeting $S_x$.
Layer cake gives
\[
  q_r(\mathcal P)=|\mathcal P|+
    \int_0^\infty N_{\mathcal P}(x)\,dx.
\]
Prefix consistency implies $N_{\mathcal P}(x)=A(S_x)$ almost everywhere,
while every feasible phase packing uses at least $\BP(S_x)$ bins on $S_x$.
Thus $q_r(\mathcal P)\le c q_r^*$.

Using the absolute $5/3$ ordinary online bin-packing algorithm
of~\cite{balogh2019} yields polynomial global ratios
\[
  \frac53\cdot\frac{3+\sqrt5}{2}
  =\frac{5(3+\sqrt5)}6
\]
deterministically and
\[
  \frac53\left(1+\frac1{\sqrt2}\right)
  =\frac{5(2+\sqrt2)}6
\]
randomly.  The latter is superseded by
Theorem~\ref{thm:main-randomized}, but the layer-cake reduction remains useful
for offline approximation and for other weighted phase objectives.

\section{A Weighted Endpoint AFPTAS for the Phase Oracle}
\label{sec:weighted-afptas}

The phase bounds of Appendix~\ref{sec:phase} are stated first with exact
endpoint optimization, but that local problem is strongly NP-hard.  This
section closes the resulting phase oracle: for every fixed phase length,
weighted endpoint packing has an AFPTAS.  It is not a separate online
paradigm, and its per-phase additive term prevents it from automatically
turning the unrestricted phase ratio into a polynomial-time global
competitive guarantee.

We now treat the finite optimization problem used inside one phase.  We fix a
constant $L>0$.  Item $i$ has size $s_i\in(0,1]$ and age
$w_i\in[0,L]$; a bin $B$ costs
\[
  c(B)=1+\max_{i\in B}w_i.
\]
This problem contains ordinary bin packing when all ages are zero, and therefore an
absolute PTAS is impossible unless $\mathrm P=\mathrm{NP}$.  The ordered ages nevertheless permit an AFPTAS.

\begin{theorem}
  \label{thm:weighted-afptas}
  For every fixed rational $L>0$, every rationally encoded instance, and every
  rational $\varepsilon\in(0,1]$, weighted endpoint packing has an algorithm
  running in time polynomial in the input length and $1/\varepsilon$ and
  returning a packing of cost
  \[
    \ALG\le(1+\varepsilon)\OPT+P_L(1/\varepsilon),
  \]
  where $P_L$ is a polynomial independent of the instance.  The empty
  instance is returned unchanged.
\end{theorem}

\begin{remark}
For a fixed explicitly represented algebraic $L$ of bounded degree and
height, the same argument works in its exact ordered field.  We keep the
formal theorem rational to make the finite-input arithmetic model explicit;
fixed rational phase lengths can approximate an algebraic target constant
arbitrarily closely.
\end{remark}

The proof combines age-stratified linear grouping with a configuration LP
whose residual-capacity constraints are nested by age.  The complete
algorithm is summarized in Algorithm~\ref{alg:weighted-afptas}.

\begin{algorithm}[t]
  \caption{Weighted endpoint AFPTAS}
  \label{alg:weighted-afptas}
  \begin{algorithmic}[1]
    \Require Items $(s_i,w_i)$, fixed $L$, accuracy $\varepsilon$
    \State Set $\tau=\min\{1/4,\varepsilon/[20(1+L)]\}$
    \State Round every age upward on a grid of width at most $\tau$ ending at $L$
    \State Call $i$ large if $s_i\ge\tau$ and small otherwise
    \For{each rounded age level}
      \State sort its large items by nonincreasing size
      \State split them into groups of size $\lceil\tau^2N\rceil$
      \State remove the first group and round every later group upward
    \EndFor
    \State approximately solve LP~\eqref{eq:weighted-config-lp} by
      knapsack pricing; recover an exactly feasible sparse primal
    \State round every positive configuration variable upward
    \State inject the retained large items into the resulting typed slots
    \State process small items from oldest to youngest; discard any tested
      residual smaller than the current item; after all eligible residuals
      are exhausted, mark all remaining items of that level as overflow
    \State pack overflow items by Next Fit in nonincreasing age order
    \State pack each removed first-group item alone; delete empty template
      bins and downgrade every remaining declared age to its oldest item
  \end{algorithmic}
\end{algorithm}

\subsection{Rounding ages and large sizes}

We use the grid $0,\tau,2\tau,\ldots,\lfloor L/\tau\rfloor\tau,L$, deleting a
duplicate final point if necessary, and round each age to its next grid
point.  Thus every rounded age remains at most $L$.  Rounding raises the cost
of each bin by at most $\tau$.  Since every nonempty bin costs at least one,
the rounded-age optimum $\OPT^+$ satisfies
\begin{equation}
  \label{eq:age-rounding}
  \OPT\le\OPT^+\le(1+\tau)\OPT.
\end{equation}
There are $K=O(L/\tau+1)$ rounded levels $W_1<\cdots<W_K$.

Within one level containing $N$ large items, use groups of size
$q=\max\{1,\lceil\tau^2N\rceil\}$.  We remove the first group and round every
item of group $j\ge2$ to the largest size in that group.  The rounded items
of group $j$ fit into the slots occupied by group $j-1$ in any original
packing: the injection preserves both capacity and age.  Across all levels,
the removed items cost at most
\begin{equation}
  \label{eq:discarded-groups}
  (1+L)(\tau\OPT+K)
\end{equation}
when packed as singletons.  The retained instance has only
$T=O(K/\tau^2)$ rounded large types.

\subsection{Configurations and nested residual capacity}

A configuration $C$ has a declared level $\ell(C)$ and a capacity-feasible
multiset of rounded large types whose ages are at most $W_{\ell(C)}$.  Write
\[
  c_C=1+W_{\ell(C)},\qquad
  r_C=1-\sum_t a_{Ct}s_t,
\]
where $a_{Ct}$ is the number of type-$t$ slots.  Empty configurations are
allowed.  If $n_t$ is the multiplicity of large type $t$ and $V_\ell$ is the
total small-item volume at level $\ell$, use the covering LP
\begin{equation}
  \label{eq:weighted-config-lp}
  \begin{aligned}
    \min\quad &\sum_C c_Cx_C\\
    \text{s.t.}\quad
      &\sum_C a_{Ct}x_C\ge n_t &&(t=1,\ldots,T),\\
      &\sum_{C:\,\ell(C)\ge p}r_Cx_C
        \ge\sum_{\ell\ge p}V_\ell &&(p=1,\ldots,K),\\
      &x_C\ge0.&&
  \end{aligned}
\end{equation}
Every packing induces a feasible solution, and therefore the LP is a lower bound on the
rounded retained instance.  The suffix direction is essential: an old small
item can use only a bin whose declared level is at least its own.

Figure~\ref{fig:nested-residual} illustrates the Ferrers structure.  The
highest-age demand has the smallest neighborhood; adding younger classes
reveals successively more residual capacity.  Consequently the $K$ suffix
inequalities in~\eqref{eq:weighted-config-lp} are precisely the relevant
fractional Hall conditions.

\begin{figure}[t]
  \centering
  \begin{tikzpicture}[
      x=1.15cm,y=.72cm,
      demand/.style={draw,rounded corners=1pt,fill=red!10,minimum width=.8cm,minimum height=.42cm},
      slot/.style={draw,rounded corners=1pt,fill=blue!10,minimum width=.8cm,minimum height=.42cm},
      edge/.style={gray!65,thin}
    ]
    \node[font=\small] at (0,4.8) {small-item level};
    \node[font=\small] at (5.1,4.8) {declared bin level};
    \foreach \y/\lab in {4/$K$,3/$K-1$,2/$2$,1/$1$}{
      \node[demand] (d\y) at (0,\y) {$\lab$};
      \node[slot] (s\y) at (5.1,\y) {$\lab$};
    }
    \foreach \a in {1,...,4}{
      \foreach \b in {\a,...,4}{
        \draw[edge,-{Latex[length=2mm]}] (d\a.east) -- (s\b.west);
      }
    }
    \node[align=center,font=\small] at (2.55,-.25)
      {eligibility is nested: a bin may accept\\only items no older than its declaration};
  \end{tikzpicture}
  \caption{Nested residual-capacity eligibility, drawn with level $K$ at the
  top.  Reversing the level convention reverses every prefix/suffix
  inequality.}
  \label{fig:nested-residual}
\end{figure}
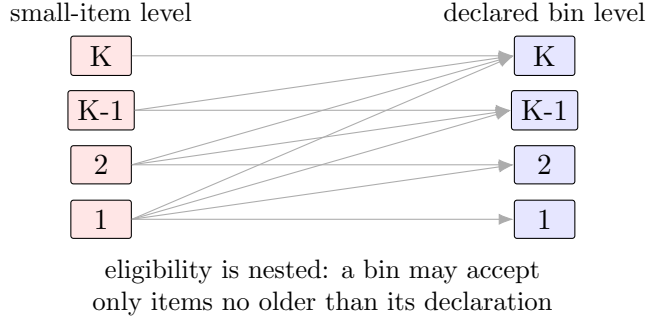

An optimal basic solution of~\eqref{eq:weighted-config-lp} uses at most
$T+K$ positive columns.  Rounding each of them upward preserves all covering
constraints and adds at most
\begin{equation}
  \label{eq:config-ceiling}
  (1+L)(T+K)=\operatorname{poly}_L(1/\varepsilon)
\end{equation}
to the objective.  Extra large-type slots are left empty; their physical
residual capacity is at least the conservative value $r_C$ used by the LP.

\subsection{Integral reconstruction of small items}

We instantiate the rounded configurations and process small items from level
$K$ down to level $1$.  We test only eligible residuals.  Whenever an item does
not fit a tested residual, discard that residual permanently; once all
eligible residuals at a level are exhausted or discarded, mark every
remaining item of that level as overflow.

Let $M$ denote the number of template bins.  After level $p$ is processed, let
$O_p$ be the overflow volume among levels at least $p$, and let $D_p$ denote the
total residual discarded from bins declared at levels at least $p$.  If no
new item overflows, downward induction gives $O_p\le D_p$.  Otherwise every
eligible original residual is either assigned or discarded.  If
$R_p^{\rm initial}$ is their initial total residual, $A_p$ the assigned
volume, and $V_{\ge p}$ the small-item demand, the suffix constraint gives
\[
  R_p^{\rm initial}=A_p+D_p,
  \qquad V_{\ge p}=A_p+O_p,
  \qquad V_{\ge p}\le R_p^{\rm initial}.
\]
Thus $O_p\le D_p$ in both cases.  Each bin discards at most one residual,
and every discarded residual is smaller than a small item, hence smaller
than $\tau$.  Therefore total overflow volume is below $\tau M$.  Empty
template bins are deleted after reconstruction.  Next Fit
uses at most
\begin{equation}
  \label{eq:small-overflow}
  \frac{\tau M}{1-\tau}+1
\end{equation}
new bins.  Processing overflow from oldest to youngest starts each bin with an oldest item.

\subsection{Polynomial pricing}

It remains to avoid explicit enumeration of all configurations.  Let $y_t$
and $z_p$ be the nonnegative dual variables for the large-type and suffix
rows.  For a fixed declared level $\ell$, put
$\Lambda_\ell=\sum_{p\le\ell}z_p$.  The maximum dual numerator of a
level-$\ell$ column is
\begin{equation}
  \label{eq:weighted-pricing}
  \Lambda_\ell+\max\left\{
    \sum_t a_t(y_t-\Lambda_\ell s_t):
    \sum_t a_ts_t\le1,
    \ w_t\le W_\ell
  \right\}.
\end{equation}
For each level, divide the returned score by its column cost
$c_\ell=1+W_\ell$, and return the largest normalized score over all $K$
levels.  Negative modified profits may be omitted.  The inner problem is
bounded knapsack: every
large type has multiplicity at most $1/\tau$.  A knapsack FPTAS
therefore gives a $(1-\rho)$ pricing oracle in time polynomial in
$1/\rho$, $1/\varepsilon$, and the encoding length~\cite{lawler1979}.

Lemma~\ref{lem:column-recovery}, proved in
Appendix~\ref{app:afptas-recovery}, then recovers a polynomial set of genuine
columns whose restricted primal is exactly feasible and has value at most
$\mathrm{LP}/(1-\rho)+\rho$.  Exact feasibility matters: it preserves every
suffix inequality used in the integral reconstruction.  The appendix spells
out the additive threshold scan, seed columns, dual bounding box, and finite
LP duality; a grey-zone binary search would not suffice.  Taking
$\rho=\tau$ yields polynomial running time.

\paragraph{End-to-end cost accounting.}
Let $\mathrm{LP}$ denote the optimum of LP~\eqref{eq:weighted-config-lp} for the
rounded retained instance.  We apply the group-shift injection simultaneously
inside an optimal rounded-age packing: retained rounded large items occupy
the preceding groups' slots and the small items stay in place.  Hence
\[
   \mathrm{LP}\le\OPT^+\le(1+\tau)\OPT.
\]
With $\rho=\tau$, Lemma~\ref{lem:column-recovery} returns an exactly feasible
basic solution of value
\[
   P\le\frac{\mathrm{LP}}{1-\tau}+\tau.
\]
It has at most $T+K$ positive variables.  After rounding them upward, let
$\overline P=\sum_Cc_C\lceil x_C\rceil$ be the declared cost of the template
bins.  Since $c_C\le1+L$,
\begin{equation}
 \label{eq:afptas-template-cost}
 \overline P\le
 \frac{\mathrm{LP}}{1-\tau}+\tau+(1+L)(T+K).
\end{equation}
If $M=\sum_C\lceil x_C\rceil$, then $M\le\overline P$ because every
configuration costs at least one.

By~\eqref{eq:small-overflow}, Next Fit opens at most
$\tau M/(1-\tau)+1$ overflow bins.  Each costs at most $1+L$, since its first
item is oldest and rounded ages do not exceed $L$.  Adding these bins and the
removed first-group singletons gives
\begin{align}
 \ALG
 &\le (1+L)(\tau\OPT+K)
 +\left(1+\frac{(1+L)\tau}{1-\tau}\right)\overline P
 +(1+L) \notag\\
 &\le\left[(1+L)\tau+
   \frac{(1+\tau)(1+L\tau)}{(1-\tau)^2}\right]\OPT
   +\Gamma_L(\tau),
 \label{eq:afptas-total}
\end{align}
where
\[
 \Gamma_L(\tau)=(1+L)K+
 \left(1+\frac{(1+L)\tau}{1-\tau}\right)
 [\tau+(1+L)(T+K)]+(1+L).
\]
Since $K=O(L/\tau+1)$ and $T=O(K/\tau^2)$, this is
$\operatorname{poly}_L(1/\tau)$.  For $0<\tau\le1/4$, put
$R(\tau)=(1+\tau)/(1-\tau)^2$.  The elementary bounds
$R(\tau)\le1+6\tau$ and $R(\tau)\le3$ give
\[
 (1+L)\tau+(1+L\tau)R(\tau)
 \le1+7\tau+4L\tau
 \le1+7(1+L)\tau.
\]
Our choice $\tau\le\varepsilon/[20(1+L)]$ therefore makes the coefficient
at most $1+\varepsilon$, while
$\Gamma_L(\tau)=\operatorname{poly}_L(1/\varepsilon)$.  This proves
Theorem~\ref{thm:weighted-afptas}.

\begin{remark}[Why this does not immediately improve the online ratio]
  Applying the AFPTAS independently in every occupied phase incurs its
  additive term once per phase.  The number of occupied phases can be linear
  in the offline optimum, and therefore these terms cannot be absorbed into one global
  additive constant.  The theorem therefore supplies a fully polynomial
  approximation for the finite weighted subproblem, but it does not by
  itself turn the exact $1+1/\sqrt2$ phase bound into a polynomial-time online
  competitive ratio.
\end{remark}

\section{An Exact Poisson Benchmark with Half-Size Items}
\label{sec:stochastic-half}

The adversarial bounds leave open how binding capacity changes the optimal
waiting rule itself.  To isolate that effect, we study the simplest
stochastic regime in which capacity is genuinely active: every item has
size $1/2$, and therefore each bin contains at most two items.  This Poisson model is an
exact binding-capacity benchmark, not a competing interpretation of the
adversarial ratios; its offline and optimal causal rates can both be solved
in closed form.  Arrivals form a rate-$\lambda$ Poisson
process on $[0,T]$, every item has size $1/2$, and the online algorithm is
not told when the input ends.  It must therefore clear a final unmatched
item using an ordinary finite timer.  We write $N_T$ for the number of
arrivals and set both costs to zero when $N_T=0$.  The two performance
statistics are
\begin{align*}
 R_{\rm RoE}(T)&=\frac{\E[\ALG_T]}{\E[\OPT_T]},\\
 R_{\rm ER}(T)&=\E\!\left[\frac{\ALG_T}{\OPT_T}
                    \mathrel{\bigg|}N_T>0\right].
\end{align*}
They generally differ for finite $T$.

\subsection{The offline matching process}

For ordered arrivals $a_1<\cdots<a_n$, replacing two singleton bins by the
pair $i<j$ saves
\[
  w(i,j)=\bigl(1-(a_j-a_i)\bigr)_+.
\]
Thus the offline problem is maximum-weight matching on the arrival points.
Although an optimum need not be unique, there is always one using only
adjacent pairs.  Indeed, uncrossing positive edges does not decrease their
total saving.  Within any remaining matched interval all relevant edges are
in the linear part of $w$, where pairing consecutive points minimizes total
edge length.  Consequently, if $C_i$ is the optimum cost on the first $i$
arrivals, then
\begin{equation}
 C_0=0,\qquad C_1=1,\qquad
 C_i=\min\{C_{i-1}+1,\ C_{i-2}+1+a_i-a_{i-1}\}.
 \label{eq:stoch-offline-dp}
\end{equation}
The unmatched, crossing, and nested exchange cases are given in
Appendix~\ref{app:stochastic-half}; the conclusion is existential because
tied instances may also have nonadjacent optima.

The recurrence also gives the exact limiting constant.  Let
$G_i=a_{i+1}-a_i$ and $W_i=(1-G_i)_+$.  If $M_i=i-C_i$ is the maximum
saving, its marginal $Z_i=M_i-M_{i-1}$ satisfies
\begin{equation}
 Z_i=(W_{i-1}-Z_{i-1})_+.
 \label{eq:stoch-saving-chain}
\end{equation}
For Poisson input, the $G_i$ are independent $\operatorname{Exp}(\lambda)$
variables.  Equation~\eqref{eq:stoch-saving-chain} resets to zero whenever
$G_i\ge1$, and therefore it has a unique stationary law and forgets its initial state
geometrically.  Its stationary distribution is particularly simple:
\begin{equation}
 \Pr(Z=0)=\frac1{\lambda+1},\qquad
 f_Z(z)=\frac{\lambda}{\lambda+1}\quad(0<z<1).
 \label{eq:stoch-stationary-law}
\end{equation}
Substitution in the distributional recursion verifies
\eqref{eq:stoch-stationary-law}, and hence
\begin{equation}
 \frac{\OPT_T}{T}\longrightarrow
 r_*(\lambda):=\frac{\lambda(\lambda+2)}{2(\lambda+1)}
 \quad\text{almost surely and in $L^1$.}
 \label{eq:stoch-opt-rate}
\end{equation}
Appendix~\ref{app:stochastic-half} supplies the invariant equation, reset
coupling, and random-index ergodic passage from items to calendar time.

\subsection{Pair immediately or use a timer}

The natural online state contains at most one pending item.  When the next
item arrives, the two are sealed immediately; otherwise the pending item is
sealed when its timer expires.  Algorithm~\ref{alg:stoch-half} uses the
rate-dependent timer sequence that is needed on a finite input without an
end signal.

\begin{algorithm}[H]
  \caption{Horizon-free pair-or-timeout for half-size Poisson input}
  \label{alg:stoch-half}
  \begin{algorithmic}[1]
    \State $k\gets0$; no item is pending
    \Statex \textbf{On an arrival with no pending item:}
    \State \quad $k\gets k+1$ and make the item pending
    \State \quad set $\tau_k\gets0$ if $\lambda\le1$; otherwise set
      $\tau_k\gets(2/\lambda)\log(k+1)$
    \State \quad start a timer of length $\tau_k$
    \Statex \textbf{On an arrival while one item is pending:}
    \State \quad cancel the timer and seal the two items together immediately
    \Statex \textbf{If the pending item's timer expires first:}
    \State \quad seal that item alone
  \end{algorithmic}
\end{algorithm}

For intuition, first fix a deterministic timeout $\tau<\infty$ and put
$q=e^{-\lambda\tau}$.  One regenerative cycle consumes $2-q$ items in
expectation and costs $1+(1-q)/\lambda$.  Its exact calendar-time cost rate is
\begin{equation}
 g_\tau(\lambda)
 =\frac{\lambda+1-e^{-\lambda\tau}}
        {2-e^{-\lambda\tau}}.
 \label{eq:stoch-fixed-timer}
\end{equation}
The derivative with respect to $q$ has the sign of $\lambda-1$.  Therefore
the best action changes sharply at rate one: seal immediately below rate one,
and make the timer diverge above rate one.

Random timers do not improve this calculation.  For any hypothetical timer
$S$, independent of the next Poisson increment conditional on the past, the
same formulas hold with $q=\E[e^{-\lambda S}\mid\text{past}]$.  More
importantly, the restriction to one pending item loses nothing.

\begin{lemma}[Immediate-pair normalization]
\label{lem:stoch-normalization}
We fix the arrival-first convention and the within-epoch identity order.  Every
admissible causal policy is pathwise dominated, on every finite labelled
input and for every fixed private seed, by an admissible causal policy with
the following rule.  At an arrival epoch, place the unique item pending from
earlier epochs, if any, before the new items, order the latter by identity,
and seal consecutive pairs immediately.  Thus at most one item remains
pending after the epoch.  The transformation need not preserve more restrictive notions of stationarity.
\end{lemma}

\begin{proof}[Proof idea]
We order pair opportunities by the global rank of their second item, using
arrival time followed by identity.  We shadow-simulate the original policy
after its first violation.  Before that violation there was at most one old
pending item, and therefore every alternative shadow partner arrived no earlier than
the violation time $t$.  If the prescribed pair would share a later shadow
bin, move that bin to $t$.  If its items would use different bins, let those
bins close at $u$ and $v$, pair the prescribed items at $t$, retain their
zero, one, or two partners, and pack the partners at
$m=\max\{u,v\}$.  With $s=\min\{u,v\}$, repaired cost minus shadow cost is
respectively
\[
 (t-\min\{a_x,a_y\})+(t-s),\qquad
 (t-a_x)+(t-s),\qquad\text{or at most }-1,
\]
and is nonpositive because every partner arrives at or after $t$ and
$s\ge t$.  At time $m$ the physical and shadow identity sets resynchronize.
Repeated first-violation repair moves the violation rank strictly forward.
Finite-history stabilization defines one projectively consistent causal
limit policy; on every finite input it agrees through cleanup with a finite
repair iterate and is no more expensive.
\end{proof}
The one-deviation operator, simultaneous-epoch convention, finite-history
stabilization, and exact identity resynchronization are formalized in
Appendix~\ref{app:stochastic-half}.

Let $\mathcal G_{k-1}$ denote the sigma-field generated by the master seed and the
complete arrival/action history through the arrival of cycle leader $k$.
Conditional on $\mathcal G_{k-1}$, let $q_k$ denote the transform of its
hypothetical timer and let $C_k,D_k$ denote cycle cost and time to the next
cycle leader.  For the two candidate rates,
\begin{align}
 \E[C_k-\lambda D_k\mid\mathcal G_{k-1}]
   &=\frac{(1-q_k)(1-\lambda)}{\lambda} &&(\lambda<1),
 \label{eq:stoch-drift-low}\\
 \E\!\left[C_k-\frac{\lambda+1}{2}D_k
      \mathrel{\bigg|}\mathcal G_{k-1}\right]
   &=\frac{q_k(\lambda-1)}{2\lambda} &&(\lambda>1).
 \label{eq:stoch-drift-high}
\end{align}
At $\lambda=1$, either display has zero drift after continuous extension.
Lemma~\ref{lem:stoch-normalization} and these nonnegative drifts give the
lower bound for every causal policy, including history-dependent and
randomized ones.
The localization and stopped-martingale argument needed to sum these drifts
through a random number of cycles is in Appendix~\ref{app:stochastic-half}.

For $\lambda>1$, an infinite timer would attain zero drift on the infinite
stream but is inadmissible on a truncated input.  Algorithm~\ref{alg:stoch-half}
uses finite timers with
$e^{-\lambda\tau_k}=(k+1)^{-2}$.  The drift errors are summable, while the
final timer is only logarithmic.  A stopped-martingale argument gives
\begin{equation}
 \E[\ALG_T]=\frac{\lambda+1}{2}T+O_\lambda(\log(T+2)),
 \qquad
 \frac{\ALG_T}{T}\longrightarrow\frac{\lambda+1}{2}
 \quad\text{in $L^1$}.
 \label{eq:stoch-alg-rate}
\end{equation}
Appendix~\ref{app:stochastic-half} proves the $L^1$ claim by a stopped
martingale estimate; the final finite timer controls the horizon boundary.
Thus the exact optimal causal long-run rate is
\begin{equation}
 g^*(\lambda)=
 \begin{cases}
   \lambda,&0<\lambda\le1,\\[2pt]
   (\lambda+1)/2,&\lambda\ge1.
 \end{cases}
 \label{eq:stoch-online-rate}
\end{equation}

\subsection{The two stochastic ratios}

Combining~\eqref{eq:stoch-opt-rate} and
\eqref{eq:stoch-online-rate} yields the complete frontier.

\begin{theorem}[Exact Poisson half-size frontier]
\label{thm:stoch-half-frontier}
For every $\lambda>0$, Algorithm~\ref{alg:stoch-half} is asymptotically
optimal among admissible causal algorithms.  Moreover,
\[
 \lim_{T\to\infty}R_{\rm RoE}(T)
 =\lim_{T\to\infty}R_{\rm ER}(T)
 =R^*(\lambda),
\]
where
\begin{equation}
 R^*(\lambda)=
 \begin{cases}
 \displaystyle\frac{2(\lambda+1)}{\lambda+2},
       &0<\lambda\le1,\\[7pt]
 \displaystyle\frac{(\lambda+1)^2}{\lambda(\lambda+2)},
       &\lambda\ge1.
 \end{cases}
 \label{eq:stoch-ratio}
\end{equation}
The function tends to one in both the sparse and dense limits and has maximum
$4/3$ at $\lambda=1$.
\end{theorem}

\begin{proof}
The rate lower bound follows from Lemma~\ref{lem:stoch-normalization} and
\eqref{eq:stoch-drift-low}--\eqref{eq:stoch-drift-high}; the matching upper
rates follow from Algorithm~\ref{alg:stoch-half} and
\eqref{eq:stoch-alg-rate}.  The ratio of expectations follows from $L^1$ rate convergence.

For the expected realized ratio, the only additional issue is uniform
integrability.  On every nonempty realization, the algorithm satisfies the
pathwise bound
\[
 \ALG_T\le N_T+T+\tau_{N_T+1},
 \qquad \OPT_T\ge N_T/2.
\]
For $N_T\sim\operatorname{Pois}(\lambda T)$, splitting at
$N_T=\lambda T/2$ and using a Poisson lower-tail bound gives
\[
 \sup_{T\text{ large}}
 \E\!\left[(T/N_T)^2\mid N_T>0\right]<\infty.
\]
Since $\tau_n=O_\lambda(\log(n+1))$, the conditional ratios are uniformly
integrable.  Their convergence in probability to $g^*(\lambda)/r_*(\lambda)$
therefore passes to conditional expectations and gives
\eqref{eq:stoch-ratio}.
The cycle-by-cycle telescope behind the pathwise envelope and the complete
conditional $L^2$ estimate are recorded in
Appendix~\ref{app:stochastic-half}.
\end{proof}

Thus the half-size model supplies the capacity-binding counterpart to the
exact nonbinding frontier of Section~\ref{sec:lower-bounds}: the performance
criteria differ, but both benchmarks identify an optimal waiting rule rather
than merely bounding one prescribed algorithm.

\begin{figure}[H]
  \centering
  \begin{tikzpicture}[x=1.25cm,y=5.8cm]
    \draw[-{Latex[length=2mm]}] (0,1) -- (5.25,1)
      node[right] {$\lambda$};
    \draw[-{Latex[length=2mm]}] (0,1) -- (0,1.38)
      node[above] {$R^*(\lambda)$};
    \draw[blue!70!black,very thick,domain=0.03:1,samples=80]
      plot (\x,{2*(\x+1)/(\x+2)});
    \draw[blue!70!black,very thick,domain=1:5,samples=120]
      plot (\x,{(\x+1)*(\x+1)/(\x*(\x+2))});
    \draw[densely dashed] (1,1) -- (1,1.3333);
    \fill[blue!70!black] (1,1.3333) circle (1.3pt);
    \foreach \x in {1,2,3,4,5}
      \draw (\x,0.995) -- (\x,1.005) node[below=2pt] {\x};
    \draw (0,1.3333) -- (-0.05,1.3333)
      node[left] {$4/3$};
    \node[blue!70!black] at (3.7,1.19) {exact causal frontier};
  \end{tikzpicture}
  \caption{The optimal long-run stochastic ratio.  The hardest arrival rate
  is the critical point $\lambda=1$; batching becomes asymptotically free in
  both the sparse and dense limits.}
  \label{fig:stoch-half-ratio}
\end{figure}
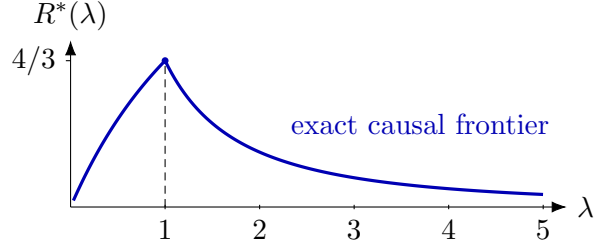

\section{Offline Complexity and Algorithmic Consequences}
\label{app:offline-computation}

\subsection{A deterministic offline two-approximation}\label{sec:offline-derand}

Theorem~\ref{thm:main-randomized} bounds the average cost over
one threshold $X$. Offline we can enumerate the possible clusterings.
Run CORS for $X=0$ and for each distinct consecutive gap in $(0,1]$.
In every run, move each seal back to the last arrival in its bin and
return the cheapest resulting partition. For any $x\in[0,1]$, let $x_0$
be the largest candidate at most $x$. A gap is cut exactly when it is
strictly larger than $x$, so $x$ and $x_0$ give identical clusters,
Next-Fit bins and normalized costs, including threshold ties. The
minimum normalized cost is therefore at most the mean online cost,
which is at most $2\OPT$. There are at most $n$ candidates and each
scan takes $O(n)$ rational arithmetic operations after sorting.

\subsection{Complexity and approximation}

If all items arrive simultaneously, every span is zero and
\eqref{eq:offline-partition} is ordinary bin packing.  Strong NP-hardness
therefore follows from \emph{3-Partition}; an absolute approximation ratio
strictly below $3/2$ would contradict the corresponding ordinary
bin-packing threshold unless $\mathrm P=\mathrm{NP}$.
Appendix~\ref{app:offline-phase} gives the explicit simultaneous-arrival
reductions, including the absolute-approximation gap between two and three bins.

The silence-threshold enumeration in Section~\ref{sec:offline-derand}
gives ratio $2$. Independently, the randomized phase construction in Appendix~\ref{sec:phase}, combined with
the absolute $5/3$ ordinary online bin-packing algorithm, can be derandomized
offline by enumerating the finitely many phase-membership intervals of the
shift.  After normalizing every produced bin to its own last arrival, this
gives a deterministic polynomial approximation ratio
\[
  \frac53\left(1+\frac1{\sqrt2}\right)
  =\frac{5(2+\sqrt2)}6
  \approx2.845178.
\]
The shift-cell enumeration and exact-arithmetic derandomization are proved in
Appendix~\ref{app:offline-phase}.

\subsection{Exact special cases}

For every nonempty input whose total size is at most one, an optimum
partitions the sorted arrivals
at every gap greater than one and nowhere else, and therefore
\begin{equation}
  \label{eq:one-bin-opt}
  \OPT(I)=1+\sum_{j=1}^{n-1}\min\{g_j,1\}.
\end{equation}
For equal item sizes, an uncrossing argument gives an optimum made of
consecutive arrival blocks; the resulting recurrence has a monotone-window
implementation running in $O(n)$ time after sorting, and hence in
$O(n\log n)$ total time.  If no feasible bin contains three
items, a pair $(i,j)$ creates saving
$\max\{0,1-|a_i-a_j|\}$ relative to two singletons.  The problem is therefore
maximum-weight matching.  For arbitrary instances, a labelled-subset dynamic
program evaluates~\eqref{eq:offline-partition} in $O(3^n)$ time and
$O(2^n)$ space.
Complete exchange arguments and recurrences for these special cases appear
in Appendix~\ref{app:offline-phase}.

\section{Deferred Proofs for Offline Structure and Phase Algorithms}
\label{app:offline-phase}

\subsection{Gap decomposition}

We fix a cut of length $g_j$.  For one crossing bin, let the extreme arrival
times of its nonempty left and right restrictions be
$\ell\le u\le v\le w$.  Splitting the bin at the cut changes its cost by
\[
 (1+u-\ell)+(1+w-v)-(1+w-\ell)=1+u-v\le1-g_j.
\]
Thus repeated splitting constructs an optimum avoiding every cut with
$g_j\ge1$.  When $g_j>1$, splitting strictly improves every crossing bin, and therefore
no optimum crosses.  Equality is sharp: two compatible items at times zero
and one cost two either together or separately.

\subsection{Hardness and the absolute barrier}

We set every arrival time to zero.  Then every feasible bin costs exactly one, and therefore the objective is ordinary bin packing without any change to the item
encoding.  The standard strongly NP-complete 3-Partition decision
instances therefore prove strong NP-hardness.

For the absolute approximation barrier, first reject the trivial no-instance
if some element exceeds half the total sum.  Otherwise scale an ordinary
Partition instance by its half-sum, and therefore every item lies in $(0,1]$, the total
size is two, and all items arrive at time zero.  Its
optimum is two if and only if the multiset can be split into two unit bins;
otherwise it is at least three.  A polynomial algorithm of fixed absolute
ratio $\rho<3/2$ would return fewer than three, hence exactly two, on every
yes-instance and at least three on every no-instance.  It would decide
Partition in polynomial time.  This argument is absolute and does
not rule out an asymptotic scheme with an additive term.

\subsection{Derandomizing the offline shifted-grid approximation}

We run the polynomial age-ordered phase packer of Appendix~\ref{sec:phase} with
$L=\sqrt2$.  Under a uniform global offset, its expected endpoint cost is at
most
\[
 \frac53\left(1+\frac1{\sqrt2}\right)\OPT.
\]
After the packing is produced, serve every output bin at its own last
arrival.  This preserves its item set and capacity and can only decrease
cost, and therefore the same expectation bounds the resulting offline temporal-span
partition.

As the offset moves on the circle of circumference $L$, phase membership
changes only at the $O(n)$ residues $a_i\bmod L$.  Inside each open residue
interval, phase membership, stable within-phase age order, the deterministic
ordinary online packing, and the normalized output partition are constant.
Enumerate one representative of every circular interval and return the
cheapest normalized partition.  Its cost is no larger than the offset
average.  For rational arrivals, comparisons with $L=\sqrt2$ are exact in
the fixed ordered field $\mathbb Q(\sqrt2)$ and have polynomial bit
complexity.  This proves the offline factor stated in
Section~\ref{sec:offline}.

\subsection{Exact special cases}

\paragraph{Nonbinding capacity.}
If total size is at most one, two bins with overlapping temporal hulls may be
merged: one setup disappears and the union span is no larger than the sum of
the old spans.  Hence an optimum consists of temporally disjoint consecutive
blocks.  Relative to one block, cutting gap $g_j$ adds one setup and removes
exactly $g_j$ span.  The cuts are therefore chosen independently at
$g_j>1$, with indifference at equality, proving
\eqref{eq:one-bin-opt}.  The empty instance has value zero.

\paragraph{Equal sizes.}
Let every item have size $s$ and $k=\lfloor1/s\rfloor$.  For two bins of
cardinalities $p,q$, sort their union and give the first $p$ occurrences to
one bin and the last $q$ to the other.  Feasibility is preserved.  If the old
hulls overlap, their total length is at least the hull of their union, while
the two sorted hull lengths sum to no more than that union length; if the old
hulls are disjoint, the operation only orders them.

For a formal induction, among the unfrozen items let $B$ denote the bin containing
the earliest occurrence and put $p=|B|$.  Exchange $B$ in turn with every
other unfrozen bin, always returning the earliest $p$ occurrences of the
two-bin union to $B$.  Each exchange preserves cardinalities and does not
increase total span.  After the sweep, $B$ contains exactly the earliest $p$
unfrozen occurrences.  Freeze it and apply the same argument to the
remainder.  The resulting optimal blocks are consecutive.

For sorted arrivals, its recurrence is
\[
 D[0]=0,\qquad
 D[i]=\min_{1\le q\le\min\{k,i\}}
       \{D[i-q]+1+a_i-a_{i-q+1}\}.
\]
Writing $j=i-q$ turns this into a sliding minimum of $D[j]-a_{j+1}$ over
$\max\{0,i-k\}\le j\le i-1$.  A monotone deque evaluates all states in
$O(n)$ time after sorting and uses $O(k)$ space.

\paragraph{Cardinality at most two.}
Build a graph on the labelled items, with an edge $ij$ exactly when
$s_i+s_j\le1$, and weight it by
$w_{ij}=(1-|a_i-a_j|)_+$.  Starting from all singleton bins, selecting a
feasible pair saves exactly $w_{ij}$; negative-saving pairs are never useful.
If no feasible bin contains three items, every packing is a matching plus
singletons and
\[
 \OPT=n-\max\{w(M):M\text{ is a matching}\}.
\]

\paragraph{Arbitrary instances.}
For a mask $M$ of labelled items, let $p(M)$ denote its least labelled member and
define
\[
 F(\varnothing)=0,\qquad
 F(M)=\min_{\substack{S\subseteq M:\ p(M)\in S\\ \vol(S)\le1}}
       \{1+\spanof(S)+F(M\setminus S)\}.
\]
Fixing $p(M)$ avoids redundant choices of the first block.  There are
$O(3^n)$ state--submask pairs and $O(2^n)$ states, proving the claimed time
and space bounds.

\section{Deferred Proofs for the Nonbinding-Capacity Frontier}
\label{app:lower-bounds}

\subsection{Deterministic lower bound}

We fix $N$ and give every item size $1/N$.  After applying
Lemma~\ref{lem:one-bin-normalization}, release the first item at time zero.
If item $i$ arrives at $a_i$ and the normalized deterministic policy would
serve it, in the absence of another arrival, at $t_i=a_i+x_i$, release item
$i+1$ at $a_{i+1}=t_i+\delta_i$.  We choose $\delta_i>0$ with
$\Delta=\sum_{i<N}\delta_i$ arbitrarily small.  Admissibility makes every
$x_i$ finite.  Since the policy is deterministic, recursively simulating
these actions compiles one fixed finite input depending on that policy.

Exactly one item is pending at every service.  We put
$S=\sum_{i<N}\min\{x_i,1\}$.  Then
\[
 \ALG\ge N+S,
 \qquad
 \OPT\le1+S+\Delta.
\]
Since $0\le S\le N-1$ and we may take $\Delta<N-1$,
\[
 \frac{\ALG}{\OPT}
 \ge\frac{N+S}{1+S+\Delta}
 \ge\frac{2N-1}{N+\Delta}.
\]
Taking $N\to\infty$ and then $\Delta\downarrow0$ proves the lower bound two.

\subsection{Finite-support randomized lower bound}

We use the distribution defined in Section~\ref{sec:lower-bounds}.  We put
$s_m=(1-1/m)^m$ and $q_m=1-s_m$.  For a deterministic clear-all policy,
charge a nonlast item in an online batch the gap to the next item and charge
the last item one plus its residual wait.  The charges sum exactly to the
online cost.  Conditional on a generated nonterminal item, let $x$ denote the
age at which the policy would serve if no new item arrived.  Its conditional
expected charge is
\[
 L(x)=\sum_{j/m\le x}p_j\frac jm
      +\left(s_m+\sum_{j/m>x}p_j\right)(x+1).
\]
On every open interval between grid points, $L$ has positive slope.  Let
$x_j=j/m$ and $r=1-1/m$.  At a grid point, arrival-first processing places
the equality gap in the first sum; since
$s_m+\sum_{\ell>j}p_\ell=r^j$, the exact cancellation is
\[
 L(x_j)-L(x_{j-1})
 =\frac1m r^j-\left(1-\frac1m\right)p_j
 =\frac1m r^j-r\frac1m r^{j-1}=0.
\]
Since $L(0)=1$, we have $L(x)\ge1$ for every $x\ge0$.  The forced $K$th item
also has charge at least one.

We define
\[
 A_{m,K}=\sum_{i=0}^{K-1}q_m^i,
 \qquad
 B_{m,K}=q_m\sum_{i=0}^{K-2}q_m^i+q_m^{K-1}.
\]
Then every deterministic policy satisfies $\E\ALG\ge A_{m,K}$.  Moreover,
$\sum_jp_j(j/m)=1-2s_m=2q_m-1$, and therefore the exact nonbinding formula gives
\[
 \E\OPT
 =1+(2q_m-1)\sum_{i=0}^{K-2}q_m^i
 =B_{m,K}.
\]
First let $K\to\infty$ and then $m\to\infty$.  Since
$q_m\to1-e^{-1}$, the ratio $A_{m,K}/B_{m,K}$ tends to $e/(e-1)$.

Finally condition any randomized policy on its private seed.  The
normalization gives
\[
 \ALG_{A_\omega}(I)\ge \ALG_{A'_\omega}(I)
\]
for every fixed seed $\omega$ and every input in the support.  The
deterministic inequality holds for $A'_\omega$, hence also after averaging.
Since the distribution has finite support for fixed $m,K$, at least one
fixed input $I$ in that support satisfies
\[
 \E_A\ALG_A(I)\ge\frac{A_{m,K}}{B_{m,K}}\OPT(I).
\]
The witness is fixed independently of the seed.  Taking the two limits in the
stated order proves the oblivious lower bound.

\section{Approximate Pricing and Exact Column Recovery}
\label{app:afptas-recovery}

This appendix isolates the nonstandard interface used in
Theorem~\ref{thm:weighted-afptas}.  After deleting zero-demand rows, write the
covering configuration LP as
\[
 \min\{c^Tx:Ax\ge b, x\ge0\}
\]
and normalize a column's dual score by its positive cost.  The pricing oracle
returns a genuine column whose normalized score is at least
$(1-\rho)$ times the true maximum score.

\begin{lemma}[Approximate pricing with exact primal recovery]
\label{lem:column-recovery}
Let $b\in\mathbb Q_{>0}^m$, and suppose the finite implicit column family has
nonnegative rational columns $A_C$, positive rational costs $c_C$, and
polynomial encoding length.  Assume that we have (i) a polynomial-size seed
column set supporting an exactly feasible solution of cost at most $U_0$,
(ii) a $(1-\rho)$ approximate normalized-pricing oracle for
$0<\rho\le1/4$, and (iii) a polynomial bound on
$\lceil U_0/\rho\rceil$ in the stated input and accuracy parameters.  Then a
polynomial list of genuine columns can be found whose restricted covering LP
is exactly feasible and has value at most
\[
  \frac{\operatorname{LP}}{1-\rho}+\rho,
\]
where $\operatorname{LP}$ is the full optimum.  The returned solution can be
chosen basic, with at most $m$ positive variables and polynomial encoding
length.
\end{lemma}

\begin{proof}
The seed solution gives $\operatorname{LP}\le U_0$.  For an objective
threshold $\beta$, run the standard exact rational ellipsoid feasibility
procedure in the affine hull of the current system~\cite{grotschel1981};
this lower-dimensional-safe version either finds a feasible point or returns
a finite rational Farkas certificate supported on queried inequalities.  We run
it in the box
\[
 0\le y_j\le U_0/b_j,
 \qquad b^Ty\ge\beta,\qquad A_C^Ty/c_C\le1
 \quad\text{for every configuration }C.
\]
Include all seed-column inequalities from the outset.  At a queried $y$, an
oracle score above one supplies a genuine violated column and a valid cut.
If the returned score is at most one, we \emph{approximately accept}: every true score is at most $1/(1-\rho)$, so $(1-\rho)y$ is feasible for the full dual.

Consequently rejection is impossible when $\beta\le\operatorname{LP}$,
whereas acceptance is impossible when
$\beta>\operatorname{LP}/(1-\rho)$.  The oracle may behave
nonmonotonically in the grey interval, and therefore scan
\[
 N=\left\lceil U_0/\rho\right\rceil,
 \qquad
 \beta_k=k\rho,
 \qquad k=0,1,\ldots,N.
\]
If a rejection occurs, the first rejected threshold $\bar\beta$ is no larger
than the first grid point above $\operatorname{LP}/(1-\rho)$ and hence
\[
 \bar\beta\le\operatorname{LP}/(1-\rho)+\rho.
\]
If no threshold is rejected, the last threshold is at least $U_0$ and was
approximately accepted.  Thus
\[
 (1-\rho)\beta_N\le\operatorname{LP},
 \qquad
 U_0\le\beta_N\le\frac{\operatorname{LP}}{1-\rho}.
\]
The seed solution therefore already satisfies the claimed bound; return an
exact basic optimum of the seed-column LP.

Otherwise let $S$ contain the seed columns and every genuine column cut from
the rejected run.  Exact rational ellipsoid feasibility returns a finite
infeasibility certificate consisting of those genuine cuts, the box, and
$b^Ty\ge\bar\beta$.  Thus the dual restricted to $S$ has no point of value at
least $\bar\beta$ in the displayed box.  This box loses nothing: weak duality
against the seed solution gives
$b^Ty\le U_0$ for every seed-feasible dual point, and nonnegativity then
gives $y_j\le U_0/b_j$.  Finite rational LP duality therefore yields an
exactly feasible primal on $S$ of value below $\bar\beta$.  Solving this
finite LP exactly and taking a basic optimum leaves at most $m$ positive
columns.  There are $O(U_0/\rho)$ ellipsoid runs, every cut has polynomial
encoding length, and standard rational basic-solution bounds give polynomial
encoding length for the recovered primal.
\end{proof}

In LP~\eqref{eq:weighted-config-lp}, singleton columns cover each retained
large type and one empty configuration declared at the highest level covers
small volume fractionally.  Assigning each large item to its singleton and
all small volume to copies of the empty column gives
$U_0=n(1+L)$.  Since $L$ is fixed and $\rho=\tau$, both
$\lceil U_0/\rho\rceil$ and the box encodings are polynomial in the input
length and $1/\varepsilon$.

For the oracle used here, fix a declared level $\ell$ and put
$\Lambda_\ell=\sum_{p\le\ell}z_p$.  Its unnormalized numerator is
\[
 \Lambda_\ell+
 \max\left\{\sum_t a_t(y_t-\Lambda_\ell s_t):
              \sum_t a_ts_t\le1,\ w_t\le W_\ell\right\}.
\]
  Divide this level's returned score by $c_\ell=1+W_\ell$, and choose the best
  normalized score over all declared levels.  Negative modified profits may
  be omitted.  The multiplicity of a type within one configuration is at most
$1/\tau$ because every retained large size is at least $\tau$; binary
expansion reduces bounded multiplicities to zero-one knapsack with polynomial
size.  A knapsack FPTAS gives the required $(1-\rho)$ oracle.  The nonnegative
baseline $\Lambda_\ell$ preserves the same multiplicative direction.

\section{Deferred Proofs for the Poisson Half-Size Frontier}
\label{app:stochastic-half}

\subsection{Adjacent matching and the offline rate}

We delete zero-saving matching edges.  Among maximum-weight matchings, let $i$
be the leftmost matched vertex and suppose it is paired with $j>i+1$.  If
$i+1$ is unmatched, replace $\{i,j\}$ by $\{i,i+1\}$.  If $i+1$ is paired
with $k>j$, replace the crossing edges by $\{i,i+1\}$ and $\{j,k\}$; each
new edge is no longer than its corresponding old edge.  If
$i+1<k<j$, positivity of the outer edge makes the truncation inactive, and
replacing $\{i,j\},\{i+1,k\}$ by $\{i,i+1\},\{k,j\}$ changes total saving by
$2(a_k-a_{i+1})\ge0$.  Thus some optimum contains $\{i,i+1\}$; remove these
vertices and repeat.  This proves existence, not uniqueness, of an adjacent
optimum and justifies~\eqref{eq:stoch-offline-dp}.

For the stationary saving increment $Z$, let $F$ denote its cdf.  For
$0\le z<1$, the recursion~\eqref{eq:stoch-saving-chain} gives
\[
 1-F(z)=\int_{[0,1-z]}
  \bigl(1-e^{-\lambda(1-z-v)}\bigr)\,dF(v).
\]
Substitution yields the atom and density in
\eqref{eq:stoch-stationary-law}.  A gap at least one resets the chain to zero
with probability $e^{-\lambda}$ independently of the past, proving uniqueness
and geometric forgetting.  Hence
$\E Z=\lambda/[2(1+\lambda)]$.  The ergodic theorem gives the saving per item,
while $N_T/T\to\lambda$ almost surely and in $L^1$.  Random-index composition
with $\OPT_T=N_T-\sum_{i\le N_T}Z_i$ proves
\eqref{eq:stoch-opt-rate}; $0\le\OPT_T/T\le N_T/T$ supplies uniform
integrability for the $L^1$ passage.

\subsection{Policy-level immediate-pair normalization}

We fix a private seed, making the policy deterministic.  Within one simultaneous
epoch, put the unique item pending from earlier epochs first, if it exists,
and order the new identities by the convention of
Section~\ref{sec:lower-bounds}.  Consecutive pairs in this virtual scan are the
canonical pair opportunities.  We order all opportunities by the global rank
of their second item.  A violation occurs when the policy does not seal the
prescribed pair at that epoch.  Before the first violation, the canonical
rule leaves at most one item from an earlier epoch pending.  Hence, if the
first violated pair is $i,j$ at time $t$, every other item available at that
point arrived at time $t$; in particular every later shadow partner of $i$
or $j$ has arrival time at least $t$.

We define a one-violation repair $\Phi(Q)$.  It follows $Q$ until this first
violation, seals $\{i,j\}$ at $t$, and thereafter maintains a shadow
execution of $Q$ on the same arrivals and seed.  Shadow bins unrelated to
$i,j$ are copied unchanged.  If $i,j$ later share one shadow bin, omit that
bin; when it closes, the physical and shadow pending identity sets coincide.
Otherwise let their distinct shadow bins close at $u,v$, put
\[
 s=\min\{u,v\},\qquad m=\max\{u,v\},
\]
and omit both bins.  We retain their zero, one, or two partners as the shadow
actions occur, and at time $m$ pack those partners as nothing, a singleton,
or a pair.  They have all arrived by $m$, and two half-size partners fit.
The physical and shadow pending identity sets then coincide, and therefore all later
shadow actions can again be copied.

The repair is causal: at $t$ it always performs the prescribed pair, while
the relation of the shadow bins and the values $u,v$ are learned only as the
shadow execution unfolds.  Admissibility makes both shadow closure times
finite on every finite continuation.  If $i,j$ share a shadow bin, moving it
to $t$ cannot increase cost.  Otherwise relabel so that $a_i\le a_j=t$.
In the two-partner and one-partner cases, repaired cost minus shadow cost is
respectively
\[
 (t-\min\{a_x,a_y\})+(t-s)\le0,
 \qquad
 (t-a_x)+(t-s)\le0,
\]
because $a_x,a_y\ge t$ and $s\ge t$.  With no partner, one setup is removed
and delay does not increase, and therefore the improvement is at least one.  Thus
$\Phi(Q)$ is pathwise no more expensive than $Q$ and can introduce a new
violation only at a strictly larger second-item rank.

We set $Q_0=Q$ and $Q_{r+1}=\Phi(Q_r)$.  On an input of $n$ arrivals, after at
most $n$ nontrivial repairs no violation remains.  To define one policy on
the entire input tree, fix a finite causal history $h$ and complete it by
silence.  The action trace of $Q_r$ through $h$ stabilizes after finitely
many repairs.  Each $Q_r$ is causal, and therefore its trace through $h$ is independent
of the chosen continuation; consequently the stabilized trace is also
continuation-independent.  We define $\widehat Q$ on $h$ by this trace.
Compatible histories give compatible stabilized traces, and therefore $\widehat Q$ is
causal.  On every finite input it agrees through final cleanup with a finite
iterate $Q_r$, hence is admissible, satisfies the canonical rule, and has
cost at most that of $Q$.  Applying the construction seedwise proves
Lemma~\ref{lem:stoch-normalization} for randomized policies.  Under the
standard-Borel convention of Section~\ref{sec:lower-bounds}, the transformation also
preserves joint measurability.  Indeed, the first violation is selected from
the countable ordering of pair opportunities, shadow closure times are
measurable stopping times, and every operation in one repair is Borel.  Thus
every $Q_r(U,h)$ is jointly measurable.  For each finite history its action
trace is eventually constant in $r$, and therefore the stabilized trace is the pointwise
limit of measurable traces and is jointly measurable as well.

\subsection{Universal drift lower bound}

Let $\mathcal G_{k-1}$ denote the sigma-field generated by the master seed and the
complete arrival/action history through the arrival of cycle leader $k$.
Expose the policy's private randomness on the no-arrival continuation and let
$S_k$ be its hypothetical service age.  Then $S_k$ is
$\mathcal G_{k-1}$-measurable, while subsequent Poisson increments are
independent of $\mathcal G_{k-1}$.  With
$q_k=\E[e^{-\lambda S_k}\mid\mathcal G_{k-1}]$, cycle cost $C_k$ and
interleader time $D_k$ satisfy
\[
 \E[C_k\mid\mathcal G_{k-1}]=1+\frac{1-q_k}{\lambda},
 \qquad
 \E[D_k\mid\mathcal G_{k-1}]=\frac{2-q_k}{\lambda}.
\]
These identities give~\eqref{eq:stoch-drift-low} and
\eqref{eq:stoch-drift-high}.  Moreover, for fresh independent
rate-$\lambda$ exponentials $X_k,Y_k$,
$C_k\le1+X_k$ and $D_k\le X_k+Y_k$, uniformly over the timer law.  Centered
increments therefore have uniformly bounded conditional second moments.

Let $L_k$ denote the time of cycle leader $k$ on an infinite Poisson continuation,
let
\[
 \begin{aligned}
 K(T)&=\max\{k:L_k\le T\} &&\text{if this set is nonempty},\\
 K(T)&=0 &&\text{otherwise},\\
 R_T&=\sum_{k\le K(T)}C_k.
 \end{aligned}
\]
with $R_T=0$ when no leader has arrived.  At the start of cycle $k$, the
indicator $\mathbf 1\{L_k\le T\}$ is measurable.  Since $K(T)\le N_T$ and the
cycle increments have uniformly bounded conditional second moments,
localization followed by $L^2$ convergence gives, for either candidate rate
$g$,
\[
 \E\!\left[\sum_{k\le K(T)}(C_k-gD_k)\right]\ge0.
\]
On $K(T)\ge1$, interleader times telescope as
\[
 \sum_{k\le K(T)}D_k=L_{K(T)+1}-L_1\ge T-L_1.
\]
Moreover, deleting arrivals after $T$ cannot make the final cycle led by time
$T$ cheaper: if its next infinite-stream arrival would precede its timer, the
truncated policy instead waits for that timer.  Hence the finite-input cost
is at least $R_T$.  Since
\[
 T-\E[(T-L_1)\mathbf 1\{L_1\le T\}]
 =\E[\min\{L_1,T\}]\le1/\lambda,
\]
every admissible causal policy satisfies
\[
 \E\ALG_T\ge gT-g/\lambda,
 \qquad
 \liminf_{T\to\infty}\frac{\E\ALG_T}{T}\ge
 \begin{cases}
  \lambda,&\lambda\le1,\\
  (\lambda+1)/2,&\lambda\ge1.
 \end{cases}
\]

\subsection{Increasing timers and terminal cleanup}

Assume $\lambda>1$, put $g=(\lambda+1)/2$, and use
$\tau_k=(2/\lambda)\log(k+1)$.  Then $q_k=(k+1)^{-2}$ and
\[
 \delta_k:=\E[C_k-gD_k\mid\mathcal G_{k-1}]
 =\frac{\lambda-1}{2\lambda}(k+1)^{-2},
\]
hence $\sum_k\delta_k<\infty$.  We define
\[
 \xi_k=C_k-gD_k-\delta_k,
 \qquad M_0=0,
 \qquad M_n=\sum_{k\le n}\xi_k.
\]
Then $(M_n)$ is a martingale.  Predictability of
$\mathbf 1\{L_k\le T\}$, the uniform conditional second-moment bound, and
$K(T)\le N_T$ give by localization and martingale isometry
\[
 \E[M_{K(T)}]=0,
 \qquad \E[M_{K(T)}^2]=O_\lambda(T),
 \qquad \E|M_{K(T)}|=O_\lambda(\sqrt T).
\]
With the convention above,
$\sum_{k\le K(T)}D_k=L_{K(T)+1}-L_1$ also for $K(T)=0$, because both sides
are then zero.  The exact telescope is
\[
 R_T-gT=M_{K(T)}+\sum_{k\le K(T)}\delta_k
       +g\bigl(L_{K(T)+1}-T-L_1\bigr).
\]
The first-leader mean is $1/\lambda$, while the overshoot
$L_{K(T)+1}-T$ is dominated in expectation by two fresh exponential gaps.
Consequently
\[
 \E R_T=gT+O_\lambda(1),
 \qquad
 \E|R_T/T-g|=O_\lambda(T^{-1/2}).
\]
Only the last cycle changes when arrivals after $T$ are deleted, and
\[
 0\le\ALG_T-R_T\le\tau_{N_T+1}.
\]
Jensen's inequality yields
$\E\tau_{N_T+1}\le(2/\lambda)\log(\lambda T+2)$.  Hence
\eqref{eq:stoch-alg-rate} holds, including its $L^1$ assertion.  For
$\lambda\le1$, the zero timer gives $\ALG_T=N_T$ and the corresponding result
directly.  Every timer in the dense-regime sequence is finite, and therefore the policy
is valid on every finite input and never uses the horizon.

\subsection{Uniform integrability of the realized ratio}

There is at most one output bin per cycle, and therefore setup cost is at most $N_T$.
Before the last cycle, each delay is at most the gap to the next cycle leader;
these gaps telescope to at most $T$.  Only the final residual item can wait
after $T$, and its delay is at most $\tau_{N_T+1}$.  Thus
\[
 \ALG_T\le N_T+T+\tau_{N_T+1},
 \qquad \OPT_T\ge N_T/2\quad(N_T>0).
\]
For $N_T\sim\operatorname{Pois}(\mu)$, $\mu=\lambda T$, split at
$N_T=\mu/2$.  Above the split, $T/N_T\le2/\lambda$; below it,
$T/N_T\le T$ and a Poisson Chernoff bound gives probability at most
$e^{-\mu/8}$.  Since $\Pr(N_T>0)\to1$,
\[
 \sup_{T\text{ large}}\E[(T/N_T)^2\mid N_T>0]<\infty.
\]
Also $\tau_{n+1}/n$ is bounded for integers $n\ge1$.  The displayed
pathwise envelope therefore gives a uniform conditional $L^2$ bound for
$\ALG_T/\OPT_T$.  Joint rate convergence yields convergence in probability
of the sample ratio, and uniform integrability upgrades it to the conditional
expectation in Theorem~\ref{thm:stoch-half-frontier}.  The two separate
$L^1$ cost-rate limits give the ratio of expectations.

\section{Exact Computation and Falsification Framework}
\label{sec:computation}

Our proofs are analytic; computation is used to discover witnesses and reject
false structural claims.  The released scripts and exact regression command
are documented in the accompanying reproducibility manifest.  The offline
oracle enumerates the subset assigned
to the first remaining labelled item and memoizes the residual mask.  It
computes~\eqref{eq:offline-partition} exactly in $O(3^n)$ time and reconstructs
an optimal partition.

The finite-game solver works on a rational time grid.  At every epoch, the
adversary reveals a complete arrival batch before the algorithm selects an
arbitrary pending subset and an arbitrary feasible partition of that subset.
Its silent mode supplies no stop signal: termination is represented by empty
epochs until the algorithm has served every item.  The finite-to-continuous
transfer is quantitative.  Normalize the grid spacing to one and scale the
setup cost to $D>0$.  We fix a seed and maintain a virtual execution of the
continuous policy on the same grid-timed arrivals.  After revealing the full
batch at epoch $k$, simulate the virtual execution through $k$, feeding that
batch to the shadow only at virtual time $k$, and combine
all as-yet-unrealized nonempty virtual services with times in $(k-1,k]$ into
one finite ordered grid trace, represented by one partition-valued
macro-action in the solver and using the disjoint union of their original bin
partitions.
Every item remains physically pending until its virtual bin is copied.  This
is a causal shadow construction, including when several services round to the
same epoch.  Each virtual bin is served once at the earliest grid epoch weakly
after its virtual service time.  A tied endpoint may reveal extra items but
does not invalidate the copied subset.  Thus every bin cost increases by at
most one; because its original cost is at least $D$, the rounded cost is at
most $(1+1/D)$ times the original.  For randomized policies the construction
also preserves joint measurability: ceiling a stopping time to the integer
grid and grouping a finite labelled action trace are measurable maps.

To prevent a finite horizon from leaking an end signal, suppose the game
releases at most $N$ items and no item after time $H$.  For a target ratio
$r$, choose a grid-compatible $L$ and append public empty epochs through
$H+L$, where
\[
 L>(rN-1)D.
\]
If, after the action at epoch $H+L$, an item is still pending, then it is
either never served, violating admissibility, or its eventual bin costs at
least $D+L>rND$, while immediate singleton service gives $\OPT\le ND$.
This nonempty terminal state is certified analytically at ratio $r$; the
declaration is not an event revealed to the online policy.  Otherwise the
execution completes inside the represented silent tail.  Hence, for
deterministic policies, a certified adaptive-grid lower bound $r$ yields the
continuous-time lower bound $r/(1+1/D)$, with no termination signal supplied
to the algorithm.  The rounding construction is seedwise, but an
oblivious-randomized lower bound additionally requires one fixed finite input
distribution and averaging over the independent seed.  The adaptive pure
minimax value reported by the finite-game solver is not such a certificate;
the separate finite-support Yao solver is the appropriate computational route
for that quantifier order.

The same tools rejected several tempting algorithms.  In particular, the
oldest-age flush-all rule has ratio tending to $5/2$ on items
\[
  (1,0),\quad(1/2,1),\quad(1/2,1+\varepsilon).
\]
A more elaborate oldest-frontier lexicographic rule has a six-item family
whose ratio tends to $8/3$.  These failures illustrate why a proof may not
assume that all currently pending items should be packed at one service.

For Algorithm~\ref{alg:cors}, the implementation integrates the
uniform-threshold cost exactly between consecutive gap breakpoints. It
also compares the offline threshold enumeration and discrete-threshold
costs with exact OPT, and checks the exponential law's temporal identity
numerically. These finite checks exercise ties and the pathwise cost
decomposition used in Theorem~\ref{thm:main-randomized}.
The ancillary directory \texttt{anc/experiments/} contains a
self-contained script, \texttt{check\_implementation.py}, and its
125 recorded cases. It compares the deterministic threshold in
$\mathbb Q(\sqrt3)$ with the same exact offline oracle.

\section{Extension to top-$k$ waiting costs}
\label{sec:top-k}

Fix an integer $k\ge1$. A bin $B$ sealed at time $t$ now pays one
plus the sum of its largest $\min\{k,|B|\}$ waiting times. Denote this
sum by $\operatorname{Top}_k(B,t)$ and sum it over bins to obtain $D^{(k)}$.
Items remain distinct at equal arrival times, and capacity and completion
requirements are unchanged. The deterministic buffer method extends with
a constant independent of $k$.

\subsection{Half-full or top-$k$}

Maintain one buffer of small items, of load strictly below $1/2$.
Process all items arriving at a timestamp in a fixed order. An item of
size at least $1/2$ is sealed immediately, together with the current buffer
if they fit, and alone otherwise. Insert a smaller item into the buffer;
if its load reaches $1/2$, seal the entire buffer. After these arrival
actions, seal a remaining buffer when its top-$k$ waiting sum reaches one.
Between arrivals, the same clock determines the next service. Arrival
actions have priority over a simultaneous timer.

New items have zero age and cannot increase the clock instantaneously.
The buffer is cleared as a whole, so every produced bin has delay at
most one, and every timer bin has delay exactly one. A positive-length
buffer episode starts when an item enters an empty buffer and ends at
the next clearing, possibly together with an arriving large item.
These episodes have disjoint interiors. A large singleton sealed without
the buffer has zero delay and does not end the episode.

\begin{lemma}[A buffer delay certificate]\label{lem:top-k-buffer}
For any feasible comparison schedule $O$, let $C_O$ be its number of bins.
The modified algorithm satisfies $D_A^{(k)}\le C_O+D_O^{(k)}$.
\end{lemma}
\begin{proof}
For each final bin $B$ of $O$, let $n_B(t)$ count its items that have
arrived but are not yet served. Directly integrating the waits of its
oldest $k$ items gives
\[
 D_O^{(k)}=\int_0^\infty\sum_{B\in O}\min\{k,n_B(t)\}\,dt.
\]
This identity uses the comparator's eventual partition only for analysis.

At every timestamp, reveal arrivals, perform the algorithm's arrival
actions, then the comparator's actions, and finally any algorithm timer
action. This interleaving changes neither schedule's timestamps or costs.
Mark a positive-length buffer episode when the comparator first serves
any item currently in that buffer. Each marked episode has delay at most
one. Its witness comparator action cannot mark another episode: after
the arrival actions there is only one current buffer, and no new buffer
starts after a timer at that timestamp. Thus marked episodes contribute
at most $C_O$. An episode cleared by an arrival action can remain unmarked
even if the comparator serves its items at that same endpoint; its
interior charge below is unaffected.

If an episode ends unmarked, all its pending buffer items are still
pending in $O$ throughout its interior. Partition them according to
their eventual comparator bins. Since
\[
 \min\{k,\textstyle\sum_B x_B\}
 \le\sum_B\min\{k,x_B\},
\]
the buffer's rate is at most the comparator's total rate at that time.
The interiors of these episodes are disjoint, so their total delay is
at most $D_O^{(k)}$. Adding the two charges proves the lemma.
\end{proof}

\begin{theorem}[Top-$k$ bin packing]\label{thm:top-k-four}
Half-full or top-$k$ is strictly $4$-competitive for every $k\ge1$,
independently of the number and sizes of the items.
\end{theorem}
\begin{proof}
Let $F$ and $T$ be the numbers of arrival-triggered and timer bins.
Every arrival-triggered bin has load at least $1/2$, so
$F\le2\vol(I)\le2C_O$. Timer bins each pay one unit of delay, so
$T\le D_A^{(k)}$. Therefore
\[
 \ALG=F+T+D_A^{(k)}
 \le2C_O+2D_A^{(k)}
 \le4C_O+2D_O^{(k)}\le4\operatorname{cost}(O).
\]
Every nonempty buffer has rate at least one and is cleared within one
time unit if no arrival clears it earlier. Each bin removes at least one
item, giving finite completion without an end signal. Empty input costs zero.
\end{proof}

Only the accumulated clock and a count capped at $k$ are needed for timer
updates. The old age-only timer cannot be retained: $k$ simultaneous tiny
items of total load below $1/2$ would cost $1+k$, while immediate offline
service costs one.

\subsection{Nonbinding capacity and the half-size benchmark}

If all input items together fit one bin, keep every pending item and seal
when its top-$k$ waiting sum reaches one. The preceding certificate now
gives $\ALG=2D_A^{(k)}\le2\OPT^{(k)}$. Thus the nonbinding case retains
a deterministic factor of two.

If every feasible bin contains at most two items, the offline optimum
is identical for maximum and top-$k$ delay. Indeed, each bin can be
advanced to its latest arrival, and a pair then has waits $(\Delta,0)$.
Every partition therefore has the same canonical cost under either
objective. In particular, the matching formulation remains exact.

\begin{corollary}[Poisson half-size items]\label{cor:top-k-half-size}
For the Poisson half-size model of Section~\ref{sec:stochastic-half},
the optimal causal rate, clairvoyant rate, and stochastic frontier in
Theorem~\ref{thm:stoch-half-frontier} are unchanged for every $k\ge1$.
The maximum value of the frontier remains $4/3$.
\end{corollary}
\begin{proof}
Offline equality holds pathwise by the pair observation. The optimal
causal policy in Algorithm~\ref{alg:stoch-half} seals a pair immediately
at its second arrival or seals a singleton at its timer. Every one of its
bins has at most one positive wait, so its cost also agrees pathwise.
Any competing top-$k$ policy pays at least its own maximum-delay cost,
whose expected rate is bounded below by the optimal causal rate already
proved there. The same policy attains that bound. The pathwise equalities
also retain the two limiting ratio statistics of that theorem.
\end{proof}

The last-arrival normal form and labelled-subset offline dynamic program
remain valid for general item sizes after replacing each subset's span by
its top-$k$ waiting sum at its latest arrival. The results above do not
assert the same extension for the sharper randomized sealing bound or
the weighted endpoint approximation scheme.

\section*{Disclosure of AI-Assisted Research and Writing}
The authors used OpenAI language-model systems, including Codex, extensively
and substantively as interactive assistants during this project.  These
systems contributed to proof exploration, counterexample search,
formalization, adversarial checking, literature organization, LaTeX drafting,
and editorial revision.  The human authors selected the research questions
and modeling choices, evaluated and verified the generated arguments, and
take full responsibility for every claim and for the final manuscript.

\bibliographystyle{plain}
\bibliography{references}

@article{dooly2001tcp,
  author  = {Dooly, Daniel R. and Goldman, Sally A. and Scott, Stephen D.},
  title   = {On-line analysis of the {TCP} acknowledgment delay problem},
  journal = {Journal of the ACM},
  volume  = {48},
  number  = {2},
  pages   = {243--273},
  year    = {2001},
  doi     = {10.1145/375827.375843}
}

@misc{bhore2026general,
  author        = {Bhore, Sujoy and Paw{\l}owski, Micha{\l} and Umboh, Seeun William},
  title         = {Online {TCP} Acknowledgment under General Delays},
  year          = {2026},
  eprint        = {2604.13428},
  archivePrefix = {arXiv},
  primaryClass  = {cs.DS}
}

@article{karlin2001dynamic,
  author  = {Karlin, Anna R. and Kenyon, Claire and Randall, Dana},
  title   = {Dynamic {TCP} Acknowledgment and Other Stories about
             {$e/(e-1)$}},
  journal = {Algorithmica},
  volume  = {36},
  number  = {3},
  pages   = {209--224},
  year    = {2003},
  doi     = {10.1007/s00453-003-1013-x}
}

@inproceedings{seiden2000guessing,
  author    = {Seiden, Steven S.},
  title     = {A guessing game and randomized online algorithms},
  booktitle = {Proceedings of STOC},
  pages     = {592--601},
  year      = {2000},
  doi       = {10.1145/335305.335385}
}

@article{ahlroth2013,
  author  = {Ahlroth, Lauri and Schumacher, Andr{\'e} and Orponen, Pekka},
  title   = {Online bin packing with delay and holding costs},
  journal = {Operations Research Letters},
  volume  = {41},
  number  = {1},
  pages   = {1--6},
  year    = {2013},
  doi     = {10.1016/j.orl.2012.10.006}
}

@inproceedings{azar2019clustering,
  author    = {Azar, Yossi and Emek, Yuval and van Stee, Rob and
               Vainstein, Danny},
  title     = {The Price of Clustering in Bin-Packing with Applications to
               Bin-Packing with Delays},
  booktitle = {Proceedings of SPAA},
  pages     = {1--10},
  year      = {2019},
  doi       = {10.1145/3323165.3323180}
}

@article{epstein2021,
  author  = {Epstein, Leah},
  title   = {On bin packing with clustering and bin packing with delays},
  journal = {Discrete Optimization},
  volume  = {41},
  pages   = {100647},
  year    = {2021},
  doi     = {10.1016/j.disopt.2021.100647}
}

@article{balogh2019,
  author  = {Balogh, J{\'a}nos and B{\'e}k{\'e}si, J{\'o}zsef and
             D{\'o}sa, Gy{\"o}rgy and Sgall, Ji{\v r}{\'\i} and van Stee, Rob},
  title   = {The optimal absolute ratio for online bin packing},
  journal = {Journal of Computer and System Sciences},
  volume  = {102},
  pages   = {1--17},
  year    = {2019},
  doi     = {10.1016/j.jcss.2018.11.005}
}

@article{simchilevi1994,
  author  = {Simchi-Levi, David},
  title   = {New worst-case results for the bin-packing problem},
  journal = {Naval Research Logistics},
  volume  = {41},
  number  = {4},
  pages   = {579--585},
  year    = {1994},
  doi     = {10.1002/1520-6750(199406)41:4<579::AID-NAV3220410409>3.0.CO;2-G}
}

@inproceedings{karmarkar1982,
  author    = {Karmarkar, Narendra and Karp, Richard M.},
  title     = {An Efficient Approximation Scheme for the One-Dimensional
               Bin-Packing Problem},
  booktitle = {Proceedings of the 23rd Annual Symposium on Foundations of
               Computer Science},
  pages     = {312--320},
  year      = {1982},
  doi       = {10.1109/SFCS.1982.61}
}

@article{lawler1979,
  author  = {Lawler, Eugene L.},
  title   = {Fast Approximation Algorithms for Knapsack Problems},
  journal = {Mathematics of Operations Research},
  volume  = {4},
  number  = {4},
  pages   = {339--356},
  year    = {1979},
  doi     = {10.1287/moor.4.4.339}
}

@article{grotschel1981,
  author  = {Gr{\"o}tschel, Martin and Lov{\'a}sz, L{\'a}szl{\'o} and
             Schrijver, Alexander},
  title   = {The Ellipsoid Method and Its Consequences in Combinatorial
             Optimization},
  journal = {Combinatorica},
  volume  = {1},
  number  = {2},
  pages   = {169--197},
  year    = {1981},
  doi     = {10.1007/BF02579273}
}

@misc{lu2026capacitated,
  author = {Lu, Tianhang and Ren, Runtian and Liu, Shengcai},
  title = {{Capacitated Line Aggregation with Maximum Delay}},
  year = {2026},
  howpublished = {arXiv:2610.26497},
  doi = {10.48550/arXiv.2610.26497},
  url = {https://arxiv.org/abs/2610.26497}
}

\end{document}